\documentclass[aps,prx,10pt,twocolumn,superscriptaddress,amsmath,amssymb,floatfix,longbibliography]{revtex4-2}
\usepackage[english]{babel}
\babeladjust{autoload.bcp47=on}
\usepackage{graphicx}
\usepackage{amsthm}
\usepackage[hidelinks]{hyperref}
\usepackage[dvipsnames]{xcolor}
\newcommand{\ket}[1]{\left|#1\right\rangle}
\newcommand{\bra}[1]{\left\langle#1\right|}
\newcommand{\Ftwo}{\mathbb{F}_2}
\newcommand{\e}{\mathrm{e}}

\DeclareMathOperator{\im}{im}
\DeclareMathOperator{\rank}{rank}

\DeclareMathOperator{\supp}{supp}
\DeclareMathOperator*{\argmin}{arg\,min}
\DeclareMathOperator*{\argmax}{arg\,max}
 
\theoremstyle{plain}
\newtheorem{definition}{Definition}
\newtheorem{lemma}{Lemma}
\newtheorem{theorem}{Theorem}
\newtheorem{proposition}{Proposition}

\renewenvironment{proof}[1][Proof]{\par\pushQED{\qed}\normalfont
\noindent #1.\enspace\ignorespaces}{\popQED\par}

\begin{document}
\title{Single-shot state preparation thresholds in quantum LDPC codes: soundness, confinement, and statistical mechanics}
\date{\today}

\author{Yuanchen Zhao}
\thanks{These authors contributed equally to this work.}
\affiliation{State Key Laboratory of Low Dimensional Quantum Physics, Department of Physics, Tsinghua University, Beijing, 100084, China}
\affiliation{Frontier Science Center for Quantum Information, Beijing 100184, China}

\author{Yi Yuan}
\thanks{These authors contributed equally to this work.}
\affiliation{State Key Laboratory of Low Dimensional Quantum Physics, Department of Physics, Tsinghua University, Beijing, 100084, China}
\affiliation{Frontier Science Center for Quantum Information, Beijing 100184, China}

\author{Zhengyi Han}
\affiliation{Yau Mathematical Sciences Center, Tsinghua University, Beijing 100084, China}

\author{Dong E. Liu}
\email{Corresponding to: dongeliu@mail.tsinghua.edu.cn}
\affiliation{State Key Laboratory of Low Dimensional Quantum Physics, Department of Physics, Tsinghua University, Beijing, 100084, China}
\affiliation{Frontier Science Center for Quantum Information, Beijing 100184, China}
\affiliation{Beijing Academy of Quantum Information Sciences, Beijing 100193, China}
\affiliation{Hefei National Laboratory, Hefei 230088, China}

\begin{abstract}
Single-shot error correction does not by itself establish equally direct and robust logical-state initialization. We study preparation of CSS quantum low-density parity-check code states using one noisy round of the code’s checks. We prove that linear confinement over a weight range growing at least logarithmically, together with at least logarithmic distance, suffices for a nonzero threshold under independent local-stochastic readout and data noise in the relevant CSS sector. The threshold is certified by ideal recovery of a decoder-defined logical channel. The key is a correction rule that preserves the logical class under local syndrome changes, allowing the proof to proceed without a local-stochastic bound on the preparation residual. This gives a single-block guarantee for quantum expander codes without dimension jumps or batching. Under the stronger linear-soundness hypothesis, we establish local-stochastic residual noise, enabling composition with compatible fault-tolerant gadgets. For a retained-record formulation, we also derive an exact statistical-mechanical representation of optimal final recovery at fixed preparation decoder. A decoder-dependent logical-sector order parameter characterizes asymptotically successful block decoding. Monte Carlo results for the three-dimensional toric code provide finite-size evidence for a threshold beyond both linear sufficient conditions. These results distinguish logical robustness from local-stochastic residual-noise control and identify decoder consistency as a route to single-shot initialization.
\end{abstract}
\maketitle

\section{Introduction}
\label{sec:intro}
{
Quantum error correction (QEC) protects logical information by encoding it nonlocally and repeatedly measuring stabilizer syndromes~\cite{shorSchemeReducing1995,steaneErrorCorrectingCodes1996,calderbankGoodQuantum1996,knillTheoryQuantumErrorcorrecting1997,gottesmanStabilizerCodes1997,terhalQuantumErrorCorrection2015}. Below threshold, repeated error correction can preserve encoded information for arbitrarily long computations~\cite{aharonovFaultTolerantQuantum1997,knillResilientQuantum1998,dennisTopologicalQuantumMemory2002,kitaevFaultTolerantQuantum2003}, while quantum low-density parity-check (qLDPC) codes can achieve fault tolerance with constant space overhead~\cite{gottesmanFaultTolerantQuantumComputation2014,fawziConstantOverheadQuantum2018b,breuckmannQuantumLowDensity2021,hastingsFiberBundle2021,panteleevAsymptoticallyGood2022,leverrierQuantumTannerCodes2022,dinurGoodQuantum2023,bravyiHighThreshold2024}. Fault-tolerant computation, however, must also initialize logical states, starting from unencoded physical states rather than from states already protected by the code. Here we consider measurement-based preparation, in which stabilizer generators are measured on a suitable physical input state and the measurement record determines a correction. A central question is whether such initialization can remain fault tolerant using only a single round of noisy stabilizer measurements.

With noisy stabilizer measurements, the standard approach is to repeat syndrome extraction and decode the resulting spacetime history; achieving distance-$d$ fault tolerance in surface-code schemes typically requires $O(d)$ rounds of measurements~\cite{dennisTopologicalQuantumMemory2002,raussendorfFaultTolerantQuantum2007,fowlerSurfaceCodes2012}. This suppresses measurement errors, but introduces a time cost that grows with the code distance even when the space overhead is constant~\cite{gottesmanFaultTolerantQuantumComputation2014}. Single-shot error correction avoids this repetition by exploiting code structure so that a single noisy syndrome measurement still leaves a controlled residual error~\cite{bombinSingleShotFaultTolerant2015,brownFaultTolerantErrorCorrection2016,campbellTheorySingleshotError2019,quintavalleSingleShotErrorCorrection2021,kubicaSingleShotQuantum2022,higgottImprovedSingleShot2023}. In Bombín's construction, redundancies among the measured operators provide metachecks that detect inconsistencies in noisy syndrome records~\cite{bombinSingleShotFaultTolerant2015,campbellTheorySingleshotError2019}. Campbell formulated soundness as a general sufficient condition, requiring sufficiently low-weight valid syndromes to admit corrections whose weights are controlled by the syndrome weight~\cite{campbellTheorySingleshotError2019}. Quintavalle \emph{et al.} introduced the weaker notion of confinement, which requires such control only within a prescribed range of reduced error weight, and showed that good confinement suffices for adversarial single-shot correction, while good linear confinement gives a sustainable threshold under local stochastic noise~\cite{quintavalleSingleShotErrorCorrection2021}. The distinction is consequential: quantum expander codes can exhibit linear confinement and single-shot error correction without good soundness~\cite{fawziConstantOverheadQuantum2018b,quintavalleSingleShotErrorCorrection2021}. These results, however, all concern the correction of noise acting on a state that is already encoded.

The corresponding single-shot problem for state preparation is more subtle. Single-shot error correction starts from a state that is already encoded, and the noisy syndrome is used to diagnose errors that have acted on that state. In measurement-based preparation, by contrast, the stabilizer measurement is itself part of the encoding map from an unencoded input into the code space. The ideal measurement outcome is generally random, and the correction conditioned on that outcome identifies the resulting syndrome sector with the target logical state. Readout errors therefore affect not only the diagnosis of a pre-existing error, but also the branch-dependent correction used to complete the preparation. A single-shot correction theorem for an encoded state consequently does not by itself establish fault-tolerant single-shot preparation. Code-specific constructions such as the three-dimensional gauge color code show that single-shot initialization is possible~\cite{bombinSingleShotFaultTolerant2015}, while recent work has highlighted a distinct difficulty for direct preparation: noisy initial measurements need not leave a residual error that is controlled by the observed syndrome alone~\cite{hongSingleshotPreparationHypergraph2025,xuBatchedHighRate2025}. The central question is therefore whether the structural conditions sufficient for single-shot error correction---in particular confinement---are also sufficient for state preparation, or whether the stronger control provided by soundness is required.


The difficulty is that confinement does not support the usual residual-noise argument for direct state preparation. Unlike soundness, confinement controls the correction weight only when the reduced error is already within a prescribed range~\cite{campbellTheorySingleshotError2019,quintavalleSingleShotErrorCorrection2021}; confinement alone therefore does not exclude residual syndromes whose compatible corrections are not controlled by the syndrome weight. Our key observation is that such control is stronger than what is needed to preserve the logical information. It is enough for the data decoder to assign corrections consistently under local changes of the syndrome: within a controlled region of syndrome space, if two syndromes differ by the syndrome of a single-qubit error, their assigned corrections differ by that physical error only up to a stabilizer. These relations telescope along a sequence of local errors, while linear confinement together with code distance rules out a nontrivial logical operator on any closed path that remains inside this region. Logical protection can therefore be established without first bounding the full residual-error distribution. The relevant structure is not only a property of the code, but also of how the decoder assigns corrections across nearby syndrome sectors.

\textbf{Summary of main results.—} In this work, we make this observation rigorous for one-round syndrome-measurement preparation of CSS qLDPC codes under local stochastic readout and data errors. If the relevant syndrome map has linear confinement over a range growing at least logarithmically and the corresponding distance is at least logarithmic, we construct syndrome and data decoders for which the logical failure probability vanishes below nonzero error thresholds, with an explicit bound exponentially small in the smaller of the confinement and distance scales. In particular, the theorem applies to quantum expander codes, whose single-shot error-correction property follows from linear confinement rather than soundness~\cite{fawziConstantOverheadQuantum2018b,quintavalleSingleShotErrorCorrection2021}. Under the stronger assumption of linear soundness over an extensive range, minimum-weight decoding yields the stronger conclusion that the preparation residual itself is local stochastic, allowing the usual modular composition with subsequent fault-tolerant gadgets. Thus confinement is sufficient for the logical-robustness threshold considered here, whereas soundness additionally controls the full residual-noise distribution. This single-block structural result is complementary to recent preparation schemes based on batching, enlarged higher-dimensional or resource-state constructions, and code switching~\cite{xuBatchedHighRate2025,hongSingleshotPreparationHypergraph2025,bergamaschiFaultTolerantSingleshot2025,tan2025singleshotuniversalityquantumldpc,golowich2025constantoverheadaddressablegatessingleshot}.

These theorems give sufficient conditions, but they do not determine the actual boundary of single-shot recoverability. To investigate the same question beyond the regime covered by the proofs, we derive an exact statistical-mechanical representation of the preparation problem for a fixed data decoder and maximum-likelihood final recovery. The posterior weight of each logical sector becomes a classical partition function, and a decoder-dependent generalized Wilson-loop order parameter approaches unity if and only if maximum-likelihood block decoding succeeds asymptotically. This extends the connection between QEC thresholds and disordered statistical-mechanical models~\cite{dennisTopologicalQuantumMemory2002} to the joint effect of noisy preparation measurements, subsequent data errors, and the decoder used to identify logical sectors; related mappings have been developed for finite-depth nonprojective-measurement protocols~\cite{zhu_nishimoris_2023}. We apply this formulation to the three-dimensional toric code, whose membrane geometry satisfies neither of the linear sufficient conditions above over the required growing range. Monte Carlo simulations show finite-size crossings in both the logical error rate and the logical-sector order parameter at nonzero readout noise, providing finite-size evidence that single-shot preparation can remain recoverable beyond the conditions established by our theorem.

Taken together, these results separate two notions that are usually tied together in fault-tolerance arguments: controlling the physical distribution of the residual error and protecting its logical effect. Soundness provides both, whereas confinement can already provide the latter through a consistent choice of decoder; the statistical-mechanical formulation provides a way to probe logical recoverability beyond either linear condition. This perspective identifies the data decoder, together with the code, as part of the structure governing the logical robustness of single-shot state preparation. Section~\ref{sec:preliminaries} defines the preparation protocol, noise model, and decoding problem; Sec.~\ref{sec:threshold} establishes the threshold results and their applications; Sec.~\ref{sec:smmapping} develops the statistical-mechanical mapping and numerical analysis; and Sec.~\ref{sec:conclusions} discusses implications and open questions.
}

\section{Preliminaries}
\label{sec:preliminaries}

\subsection{CSS codes and decoding}
{
A CSS code is specified by binary parity-check matrices $H_X$ and $H_Z$,
whose rows record the supports of the $X$ and $Z$ stabilizer generators.
We work over the binary field $\Ftwo$, with addition modulo two.
Commutation of the stabilizer generators requires $H_XH_Z^T=0$.
The matrix $H_Z^T$ maps a selection of $Z$ checks to the qubit support
of their product, and $H_X$ maps this support to the syndrome measured
by the $X$ checks. These consecutive maps therefore compose to zero,
which is the defining property of a chain complex.
Their transposes give the corresponding cochain complex:
\begin{equation}
\begin{aligned}
 C_2&\xrightarrow{\partial_2=H_Z^T}C_1
       \xrightarrow{\partial_1=H_X}C_0,\\
 C^2&\xleftarrow{\delta_1=H_Z}C^1
       \xleftarrow{\delta_0=H_X^T}C^0.
\end{aligned}
\end{equation}
The chosen bases identify $C_1$ and $C^1$ with the $n$ physical qubits,
$C_2$ and $C^2$ with the $r_Z$ measured $Z$ checks, and $C_0$ and
$C^0$ with the $r_X$ measured $X$ checks. With
$r_X'=\rank H_X$ and $r_Z'=\rank H_Z$, the number of logical qubits is
$k=n-r_X'-r_Z'$. We assume $k\geq1$ when defining the order parameter.
For a binary vector $e$, we write $X_e$ and $Z_e$ for the products
of the corresponding single-qubit Pauli operators over $\supp(e)$.

Logical $Z$ operators are represented by homology classes and logical
$X$ operators by cohomology classes:
\begin{equation}
\begin{aligned}
 Z_1&=\ker H_X, & B_1&=\im H_Z^T, & H_1&=Z_1/B_1,\\
 Z^1&=\ker H_Z, & B^1&=\im H_X^T, & H^1&=Z^1/B^1.
\end{aligned}
\end{equation}
In this notation, $H_Ze$ is the syndrome of an $X$ error $e$.
A valid syndrome belongs to $\im H_Z$, whereas a binary measurement
record need not belong to this space. The relations among the measured
checks are specified by $\ker H_Z^T$. Equivalently, the valid syndromes
form a classical linear code within $C^2$, whose parity
constraints are metachecks. We write $|e|$ for Hamming weight and also
use $e$ for its support when forming intersections or containment events.
The $X$ and $Z$ distance of the code is defines as
\begin{equation}
\begin{aligned}
    &d_X=\min_{l^*\in Z^1\setminus B^1}|l^*|, \qquad d_Z=\min_{l\in Z_1\setminus B_1}|l|.
\end{aligned}
\end{equation}

Consider correct $X$ data error with $Z$ check. For each valid syndrome $\sigma$, a general data decoder takes the information of syndrome configuration and returns a correction
$c(\sigma)$, represented by a binary vector on the qubits:
\begin{equation}
\label{eq:data-decoder}
 c:\im H_Z\longrightarrow C^1,\qquad
 H_Zc(\sigma)=\sigma,\qquad c(0)=0.
\end{equation}
We use $X_{c(\sigma)}$ to denote the corresponding Pauli  $X$ correction operator of $c(\sigma)$. Multiplying $X_{c(\sigma)}$ by an $X$ stabilizer has no effect on a code state,
whereas multiplying it by a logical $X$ operator can change the extension to
arbitrary logical inputs. The decoders for $Z$ data errors have the similar properties. Notice that in realistic fault-tolerance procedures, Pauli corrections can be tracked
classically in a Pauli frame rather than applied
immediately~\cite{knillQuantumComputingRealistically2005}.
}

\begin{figure*}[t]
\centering
\includegraphics[width=0.9\textwidth]{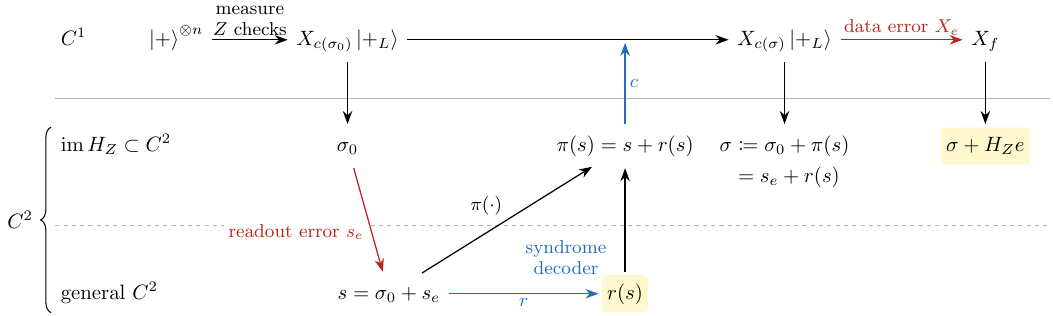}
\caption{\textbf{Single-shot state preparation and the records used for recovery.}
Time runs from left to right, and each row contains the quantities in the
indicated space. Red arrows denote errors, blue arrows denote the syndrome decoder
$r$ and the data decoder $c$, and black downward arrows point from each
state or error to its syndrome. The syndrome decoder of
Eq.~\eqref{eq:syndrome-decoder} maps the observed record $s$ to $r(s)$ and
to the valid syndrome $\pi(s)$, which determines the feedback
$X_{c(\pi(s))}$. Here $X_f$, with $f=c(\sigma)+e$, is the total error of
the channel $\mathcal N\circ\mathcal S$. The final recovery can use the
shaded records: the ideal recovery uses only the final syndrome
$\sigma+H_Ze$, and the maximum-likelihood recovery of
Sec.~\ref{sec:smmapping} also uses $r(s)$.}
\label{fig:syndrome-flow}
\end{figure*}

\subsection{Single-shot state preparation and noise model}

Since the $X$ and $Z$ sectors decouple for CSS codes, without loss of generality, we assume to begin with the $\ket{+}^{\otimes n}$ state and measure all $Z$ generators
once, and consider stochastic $Z$ readout error and $X$ data error. $Z$ data errors during this process could be corrected through the standard QEC procedure. In the absence of readout noise, the outcome $\sigma_0\in\im H_Z$
is uniformly distributed, with probability $2^{-r_Z'}$. Indeed, every nonidentity
product of independent $Z$ generators has zero expectation in the
initial product state. The outcomes of the $r_Z'$ independent generators
therefore have equal probability. Conditional on $\sigma_0$, the post-measurement state is
$X_{c(\sigma_0)}\ket{+_L}$, for an arbitrary decoder $c$, where $\ket{+_L}$ is the code state with
eigenvalue $+1$ for every logical $X$ generator.

We model readout errors by a binary random vector $S_e\in C^2$,
independent of the ideal measurement outcome, with probability
$\Pr(S_e=s_e)=q(s_e)$. The observed record is $s=\sigma_0+s_e$.
The joint state of the record and the data after measurement is therefore
\begin{equation}
\label{eq:postmeasurement}
\begin{aligned}
 \rho_Q={}&\frac{1}{2^{r_Z'}}
 \sum_{s\in C^2}\sum_{\sigma_0\in\im H_Z}q(\sigma_0+s)
 \ket{s}\bra{s}\\
 &\quad\otimes X_{c(\sigma_0)}\ket{+_L}\bra{+_L}X_{c(\sigma_0)}.
\end{aligned}
\end{equation}
Tracing out the data gives the distribution of the observed record,
\begin{equation}
\label{eq:record-law}
 \Pr(s)=\frac{1}{2^{r_Z'}}\sum_{\sigma_0\in\im H_Z}q(\sigma_0+s).
\end{equation}
The sum accounts for every valid syndrome compatible with the noisy
record, including the redundant-check constraints.

Throughout this paper, we consider the phenomenological noise model in which the readout error and data error are described separately. 
For the threshold theorem, we consider local stochastic readout noise
followed by independent local stochastic Pauli $X$ data noise.

\begin{definition}[Local stochastic error]
\label{def:local-stochastic}
A random variable of data error configuration $E$ is $(p)$-local stochastic if, for every fixed error configuration $e$,
\begin{equation}
 \Pr(e\subseteq E)\leq p^{|e|}. \label{eq:local-stochastic-error}
\end{equation}
The same definition applies to a readout error configuration $S_e$ on check bits, with
parameter $q$.
\end{definition}

The definition allows correlations within the data errors or within the
readout errors. The bound applies to the probability that every qubit
of a specified set $\supp (e)$ is noisy, whether or not there are additional errors
elsewhere. Independent Bernoulli errors are a special case in which the equality in Eq. \eqref{eq:local-stochastic-error} is taken.

\subsection{Effective error channel}

Readout errors can move the observed record outside the space of valid
syndromes, i.e. $\im H_z$. We therefore need to first decode syndrome errors before applying a data
correction.
For the case of linear soundness, i.e. part~1 of Theorem~\ref{thm:single-shot-preparation},
we choose the minimum-weight decoder:
\begin{equation}
\label{eq:syndrome-decoder}
 r_{\rm min}(s)\in\argmin_{s'\in C^2:\,s+s'\in\im H_Z}|s'|,
 \qquad \pi(s)=s+r(s).
\end{equation}
{ We choose one representative for each coset, so that
$r(s+\sigma)=r(s)$ for every valid syndrome $\sigma$.
For the case of linear confinement, i.e. Part~2 of Theorem~\ref{thm:single-shot-preparation}, we use the syndrome decoder that replaces the Hamming weight above with a variant of closeness weight \cite{quintavalleSingleShotErrorCorrection2021}, which is defined in
Appendix~\ref{app:polynomial-threshold}.}
Applying $X_{c(\pi(s))}$ to the post-measurement state changes its syndrome
to $s_e+r(s)$. Since $r(s)=r(s_e)$, the residual syndrome is the
random variable
\begin{equation}
\label{eq:residual-syndrome}
 \Sigma=\pi(S_e)=S_e+r(S_e).
\end{equation}
Its distribution is
\begin{equation}
\label{eq:residual-syndrome-law}
 \Pr(\Sigma=\sigma)
 =\frac{1}{2^{r_Z'}}\sum_{s\in C^2}q(\sigma+r(s)).
\end{equation}
Each syndrome coset contributes $2^{r_Z'}$ identical terms to the sum.
After we discard the record, the prepared state is described by the
effective error channel $\mathcal S$:
\begin{equation}
\label{eq:preparation-channel}
\begin{aligned}
 \rho_{\rm final}&=\mathcal S(\ket{+_L}\bra{+_L}),\\
 \mathcal S(\rho)&=\sum_{\sigma\in\im H_Z}\Pr(\Sigma=\sigma)
 X_{c(\sigma)}\rho X_{c(\sigma)}.
\end{aligned}
\end{equation}

Equation~\eqref{eq:preparation-channel} also defines a Pauli channel on
arbitrary encoded inputs. Let $V$ be an isometric encoding into the code
space. Then $X_{c(\sigma)}V$ encodes the same logical input into the
subspace with syndrome $\sigma$, and the corresponding noisy encoding
is $\mathcal S(V\rho_LV^\dagger)$ for a logical density matrix $\rho_L$.
Although logical $X$ operators leave $\ket{+_L}$ unchanged, the residual
syndrome can affect the correction chosen in a subsequent
error-correction cycle. We use this extension as a robustness criterion
for preparation to assess whether logical information remains protected
during the subsequent computation in the presence of preparation errors
and additional data errors. The fidelity of the prepared $\ket{+_L}$
state alone cannot establish the correctness of that computation.

In a practical implementation, the Pauli corrections determined from
measurement outcomes can be tracked classically in a Pauli frame and
propagated through the logical Clifford circuit. Residual preparation
errors must then be accounted for together with the data errors
accumulated before the next error-correction cycle. In this circuit
interpretation, $\rho_L$ denotes the ideal logical output state of the
Clifford circuit, evaluated immediately before the next error-correction
cycle. To describe the noise at this point, the preparation errors
represented by $\mathcal S$ must also be propagated through the circuit.
The effective noise therefore includes both the propagated preparation
errors and the additional data errors modeled below.

The resulting logical error rate can depend on the choice of the data
decoder $c$. To illustrate this dependence, consider the
two-dimensional toric code and choose every correction $c(\sigma)$
to avoid the support of a chosen representative of a logical $Z$
operator. Let $\sigma$ consist of two adjacent syndrome defects on
opposite sides of this logical $Z$ loop, and let $X_e$ be the
single-qubit $X$ error connecting them, so that $H_Ze=\sigma$.
The combined error $X_eX_{c(\sigma)}$ has zero syndrome and
anticommutes with the chosen logical $Z$ operator. It therefore
represents a nontrivial logical $X$ operator. Thus there could be a constant logical failure probability. 
{ Under linear soundness,
we use a minimum-weight data decoder,}
\begin{equation}
\label{eq:mw-data-decoder}
 c_{\rm min}(\sigma)\in\argmin_{x\in C^1:\,H_Zx=\sigma}|x|,
\end{equation}
{ where one minimum-weight correction is chosen for each syndrome.
For linear confinement, we use the data decoder constructed in
Appendix~\ref{app:polynomial-threshold}. It constructs corrections along paths of single-qubit syndrome changes within a region of small closeness weight, with different paths between the same endpoints giving equivalent corrections up to $X$ stabilizers.} For numerical simulations, we will use the belief propagation with
localized statistics decoding (BPLSD) decoder.

Another reason to consider the combination of $\mathcal S$ with data error is that the effective error channel $\mathcal S$ alone is exactly correctable.
To see this, write $\Pi$ for the code-space projector.
Then
\begin{equation}
\label{eq:preparation-kl}
 \Pi X_{c(\sigma)}X_{c(\tau)}\Pi
 =\delta_{\sigma,\tau}\Pi.
\end{equation}
Here $\delta_{\sigma,\tau}$ is the Kronecker delta. For distinct
syndromes, the operator between the projectors has nonzero syndrome;
for equal syndromes $\sigma$ and $\tau$, it is the identity since $c(\sigma)=c(\tau)$. The Kraus operators of
$\mathcal S$ therefore satisfy the Knill--Laflamme
conditions~\cite{knillTheoryQuantumErrorcorrecting1997}.
A perfect syndrome measurement would identify and remove the preparation
error. So the practically meaningful logical robustness must be examined under additional data noise,
\begin{equation}
\label{eq:data-noise}
 \mathcal N(\rho)=\sum_{\eta\in C^1}p(\eta)X_\eta\rho X_\eta,
 \qquad p(\eta)=\Pr(E=\eta).
\end{equation}
The total error in $\mathcal N\circ\mathcal S$ is $F=E+c(\Sigma)$.
{ To define the logical error rate, we assume an ideal final recovery
that perfectly measures all the checks and applies the correction
$c(H_ZF)$~\cite{quintavalleSingleShotErrorCorrection2021}.}

\section{State-preparation threshold theorem}
\label{sec:threshold}

{\subsection{Soundness, confinement, and the threshold bound}}

\begin{figure}[t]
\centering
\includegraphics[width=0.5\textwidth]{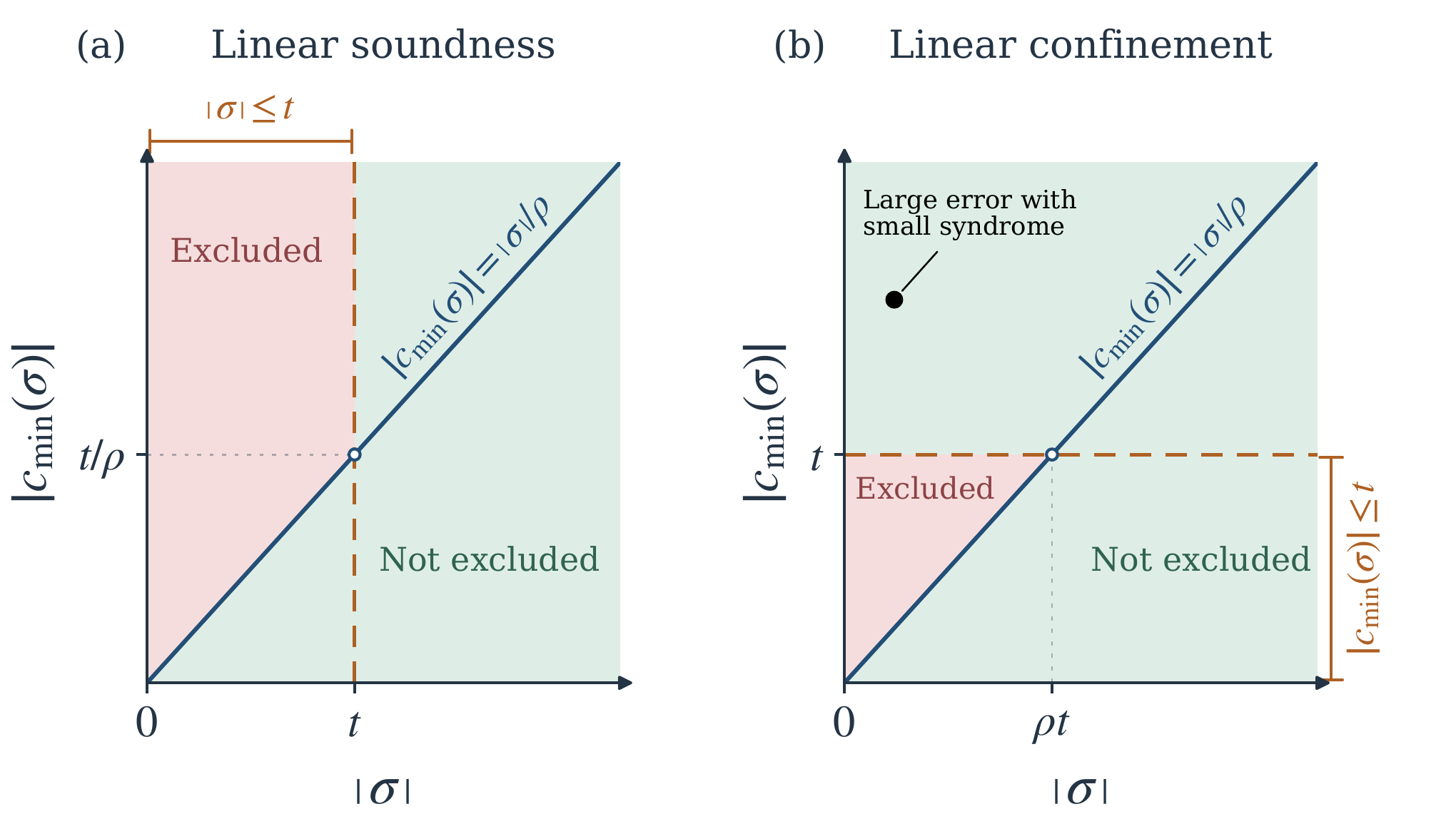}
\caption{\textbf{Schematics of linear soundness versus linear confinement.} In both diagrams, the red region is excluded by the condition, while the green region is allowed. (a) Linear soundness where $|\sigma|\leq t$ leads to $|c_{\rm min}(\sigma)|\leq f(|\sigma|)$. (b) Linear confinement where $|c_{\rm min}(\sigma)|\leq t$ leads to $|c_{\rm min}(\sigma)|\leq f(|\sigma|)$. The black spot is an example of large error with only small syndrome. However, we show that linear confinement guarantees single-shot state preparation even in the presence of such errors. As the diagrams show, ($t$,$\rho$)-linear soundness directly implies ($t/\rho$,$\rho$)-linear confinement by excluding the same region.
 }
\label{fig:compare}
\end{figure}

{Here we introduce the specific definition of soundness and confinement, and then state our main theorem.} We use the definition of
Refs.~\cite{campbellTheorySingleshotError2019,quintavalleSingleShotErrorCorrection2021},
restricted to the Pauli $X$ sector.

\begin{definition}[Soundness]
\label{def:soundness}
Let $t\geq0$ and let $f$ be a nonnegative function on the nonnegative
integers with $f(0)=0$. The matrix $H_Z$ is $(t,f)$-sound if every
$\sigma\in\im H_Z$ with $|\sigma|\leq t$ satisfies
\begin{equation}
 |c_{\rm min}(\sigma)|\leq f(|\sigma|),
\end{equation}
where $c$ is a minimum-weight data decoder. It is $(t,\rho)$-linearly sound
when $f(x)=x/\rho$ for a constant $\rho>0$.
\end{definition}

Soundness ensures that every sufficiently low-weight valid syndrome
admits a correction of controlled weight. 
For comparison, the confinement condition of
Ref.~\cite{quintavalleSingleShotErrorCorrection2021} is as follows.
\begin{definition}[Confinement]
\label{def:confinement}
Let $t\geq0$ and let $f$ be a nonnegative, nondecreasing function on the
nonnegative integers with $f(0)=0$. The matrix $H_Z$ has $(t,f)$
confinement if every $\sigma\in\im H_Z$ with $|c_{\rm min}(\sigma)|\leq t$
satisfies
\begin{equation}
\label{eq:confinement}
 |c_{\rm min}(\sigma)|\leq f(|\sigma|),
\end{equation}
where $c$ is a minimum-weight data decoder. It has $(t,\rho)$-linear
confinement when $f(x)=x/\rho$ for a constant $\rho>0$.
\end{definition}
{ Confinement gives such control only for errors whose
reduced weight is already within a prescribed range. It
does not exclude low-weight syndromes with high-weight
minimum-weight corrections. Fig. \ref{fig:compare} illustrates the difference between linear soundness and linear confinement. Soundness properties directly imply confinement, but not vice versa \cite{quintavalleSingleShotErrorCorrection2021}, as the figure also shows.
However, we will show that linear confinement suffices for single-shot state preparation. Notice that in these two definitions, $c_{\rm}$ denotes a minimum-weight
data decoder, but the preparation procedure may use a different decoder.}

Qubits that share a measured check can contribute to the same syndrome bit.
Suppose each $Z$ check acts on at most $d_c$ qubits and each qubit
participates in at most $d_v$ such checks. Define the check-adjacency graph $G(H_Z)$ as 
whose vertices are qubits, with two qubits adjacent if they appear in a
common $Z$ check. Its maximum degree is therefore at most
$\Delta=(d_c-1)d_v$. Throughout, we take $\Delta>1$. 

{
We assume the state preparation procedure is formally followed by an ideal final recovery, which means the $Z $ checks are measured perfectly, and then we apply data decoder correspondingly. 
Denote $P_{\rm fail}(p,q)$ as the probability that the final recovery leaves a nontrivial logical $X$ error.}

{
\begin{theorem}[Single-shot state-preparation threshold]
\label{thm:single-shot-preparation}
Consider a CSS qLDPC code family with $Z$-check weights and
qubit degrees bounded by $d_c$ and $d_v$, Let $\Delta=(d_c-1)d_v$, and $X$ distance $d_X=\Omega(\log n)$.
Let $S_e$ be $(q)$-local stochastic readout error and $E$ be an independent
$(p)$-local stochastic data $X$ error. The following statements hold.
\\
\noindent 1. If $H_Z$ is $(\kappa n,\rho)$-linearly sound for constants $\kappa$ and $\rho$, use the
minimum weight syndrome and data decoders.
Define
\begin{equation}
\label{eq:q-prep-theorem}
 \alpha=\min\{\rho,\kappa\},\qquad
 \widetilde q(q)=\e(2\sqrt q)^\alpha,
\end{equation}
where $\e$ is the base of the natural logarithm. If $q<1/4$ and
$\Delta\widetilde q\leq1/2$, then the residual data error obeys
\begin{equation}
\label{eq:prep-containment}
 \Pr(e\subseteq c(\Sigma))\leq\widetilde q^{|e|}
 \quad\text{for every qubit set }e.
\end{equation}
Consequently $F=E+c(\Sigma)$ is local stochastic with parameter
\begin{equation}
\label{eq:s-total-theorem}
 s(p,q)=p+\widetilde q(q).
\end{equation}
If $2\e\Delta\sqrt{s(p,q)}<1$, a final minimum-weight recovery
using a perfect syndrome measurement satisfies
\begin{equation}
\label{eq:main-failure-bound}
 P_{\rm fail}(p,q)\leq
 \frac{n}{\e\Delta}
 \frac{[2\e\Delta\sqrt{s(p,q)}]^{d_X}}
 {1-2\e\Delta\sqrt{s(p,q)}}.
\end{equation}
\\
\noindent 2. If $H_Z$ has $(t,\rho)$-linear confinement
for $t = \Omega(\log n)$ and a constant $\rho>0$ independent of $n$,
use the syndrome and data decoders constructed in
Appendix~\ref{app:polynomial-threshold}.
Define
\begin{equation}
\begin{aligned}
 \widehat q(q)&=
 \bigl[2\e\max\{2,d_c(d_v-1)\}\bigr]^{2d_v}q^{\rho/32},\\
 \widehat p(p)&=(2\e\Delta)^{2d_cd_v}p^{\rho/(16d_v)},\\
 a(p,q)&=\max\left\{
 \widehat p(p),
 \widehat q(q)
 \right\}.
\end{aligned}
\end{equation}
For the total error $F=E+c(\Sigma)$, 
If $a(p,q)<1$, a final perfect syndrome measurement
followed by the same data decoder $c$ satisfies
\begin{equation}
\label{eq:polynomial-failure-bound}
 P_{\rm fail}(p,q)\leq
 (1+d_v)n\,a(p,q)^{\left\lfloor
 \min\{t,d_X/2\}
 \right\rfloor}.
\end{equation}
The decoders do not depend on the noise distribution.
\end{theorem}
}

 Part~1 is proved in Appendix~\ref{app:threshold}.
Briefly speaking, minimum-weight decoding of syndrome errors ensures that actual readout
errors occupy at least half the support of any valid syndrome contained
in $\Sigma$.
Soundness then converts this syndrome weight into a bound on the size of
each possible residual error cluster. Counting clusters attached to an
arbitrary prescribed set gives Eq.~\eqref{eq:prep-containment}.
A second cluster argument bounds the probability that the combined error
and final recovery contain a nontrivial logical operator.

{ Part~2 is proved in Appendix~\ref{app:polynomial-threshold}.
The syndrome decoder controls the number of readout errors in connected
sets, and the data decoder is chosen so that sufficiently sparse data
errors preserve the logical information. This construction gives a
threshold when linear confinement holds up to weight $t=\Omega(\log n)$. It does not establish the local stochastic property of
residual error. Under the stronger assumption in part~1,
Eq.~\eqref{eq:prep-containment} gives this property, thus directly allowing
the preparation procedure to be composed with other fault-tolerant gadgets. The comparison of our results with previous studies on single-shot state preparation threshold is in Table \ref{tab:single-shot-comparison}.}

\subsection{Examples}

As a first example, we consider the asymptotically good quantum locally
testable CSS codes proposed recently~\cite{gay2026asymptoticallygoodquantumlocally,li2026transversalnoncliffordgatesgood}.
These codes have constant rate, linear distances, and uniformly bounded
check weights and qubit degrees. Their soundness bound gives
$|c(\sigma)|\leq|\sigma|/\rho$ for every $\sigma\in\im H_Z$, with
$\rho>0$ independent of $n$. Thus both their $Z$ and $X$ check matrices
satisfy the required $(\kappa n,\rho)$-linear soundness.
{ Part~1 of}
Theorem~\ref{thm:single-shot-preparation} therefore yields a single-shot
state-preparation threshold with exponentially small logical failure
probability at sufficiently small readout and data error rates.

We also provide another example,
using the two-skeletons of suitable finite quotients of a three-dimensional
Bruhat--Tits building \cite{bruhatTits1972,lubotzkySamuelsVishneExplicit2005}. The constructions and expansion results in
Refs.~\cite{lubotzkySamuelsVishneExplicit2005,%
evraKaufmanZemorDecodable2022,evraKaufmanBoundedDegree2017}
give a family of CSS qLDPC codes whose $Z$-check syndrome maps are
$(\kappa n,\rho)$-linearly sound for constants $\kappa,\rho>0$
independent of $n$. The code parameters satisfy
\begin{equation}
 k>0,\qquad d_X=\Omega(n),\qquad
 d_Z=\Omega(\log n).
\end{equation}
This family therefore admits a single-shot state-preparation threshold
by { part~1 of} Theorem~\ref{thm:single-shot-preparation}. The construction and
verification of the theorem's assumptions are given in
Appendix~\ref{app:ramanujan}.

In addition, the concatenated code in the Batched state preparation protocol~\cite{xuBatchedHighRate2025} serves as another example. Our theorem rigorously establishes the fault-tolerance guarantee of this protocol.

{

Quantum expander codes are hypergraph-product CSS qLDPC codes constructed
from classical expander codes~\cite{leverrierQuantumExpanderCodes2015}.
With fixed graph degrees and sufficiently strong two-sided vertex
expansion of the classical Tanner graphs, these families have
$d_X=\Omega(\sqrt{n})$, and syndrome
robustness~\cite{fawziConstantOverheadQuantum2018b,quintavalleSingleShotErrorCorrection2021}
gives
\begin{equation}
\label{eq:expander-confinement}
 |c(\sigma)|\leq\alpha\sqrt{n}
 \quad\Longrightarrow\quad
 |\sigma|\geq\rho|c(\sigma)|,
\end{equation}
for any minimum-weight correction $c(\sigma)$, with constants
$\alpha,\rho>0$ independent of $n$. Thus these families have
$(\alpha\sqrt{n},\rho)$-linear confinement in the sense of
Definition~\ref{def:confinement}. Part~2 of Theorem~\ref{thm:single-shot-preparation}, with
$t=\alpha\sqrt{n}$, therefore guarantees a nonzero single-shot state-preparation
threshold for suitable syndrome and data decoders.
}

\begin{table*}[t]
\centering
\small
\renewcommand{\arraystretch}{1.25}
\setlength{\tabcolsep}{3pt}
\newcommand{\topcell}[1]{%
  \begin{tabular}[t]{@{}c@{}}#1\end{tabular}%
}
\begin{tabular*}{\textwidth}{@{\extracolsep{\fill}}l|ccccc@{}}
\hline\hline
&
\topcell{Theorem \ref{thm:single-shot-preparation}\\Part 1}
&
\topcell{Theorem \ref{thm:single-shot-preparation}\\Part 2}
&
\topcell{Stat.-mech. mapping\\and numerics}
&
\topcell{Tan et al.\\
\cite{tan2025singleshotuniversalityquantumldpc}}
&
\topcell{Golowich et al.\\
\cite{golowich2025constantoverheadaddressablegatessingleshot}}
\\
\hline
Conditions
&
\topcell{
($\Theta(n)$,$O(1)$)-\\
Linear soundness\\
+ $d_X=\Omega(\log n)$}
&
\topcell{
($\Omega(\log n)$,$O(1)$)-\\
Linear confinement\\
+ $d_X=\Omega(\log n)$}
&
\topcell{

Example: 3D toric code \\
($\Theta(n^{1/3})$,$\Theta(x^2)$)-soundness\\
+ $d_X=\Theta(n^{2/3})$}
&
\topcell{
($\Omega(n^\beta)$,$O(1)$)-\\
Linear soundness\\
+ $d_{\rm ss}=\Omega(n^\beta)$\\
LDPC metachecks\\
+ $d_X=\Omega(n^\beta)$}
&
\topcell{
LDPC metacheck\\
$(\Omega(n^\beta),0)$-\\
small-set flip decoder\\
+ $d_X=\Omega(n^\beta)$}
\\
\hline
Results
&
\topcell{
Local-stochastic\\
residual\\
+ Exponentially small\\
logical error}
&
\topcell{
Exponentially small\\
logical error}
&
\topcell{
Numerical evidence\\
for a threshold}
&
\topcell{
Local-stochastic\\
residual (w.h.p.)\\
+ Exponentially small\\
logical error}
&
\topcell{
Exponentially small\\
logical error}
\\
\hline\hline
\end{tabular*}
\caption{Comparison of results on single-shot state-preparation thresholds for qLDPC codes.
Here $\beta>0$ is a constant, and $d_{\rm ss}$ denotes
the single-shot distance of the metachecks \cite{campbellTheorySingleshotError2019}. The abbreviation w.h.p.\ means
with high probability, where the probability of bad events is exponentially small. The small-set flip decoder in Ref. \cite{golowich2025constantoverheadaddressablegatessingleshot} repairs up to $O(n^\beta)$ syndrome errors through local flips with noiseless metasyndromes. It implies ($\Omega(n^\beta)$,$O(1)$)-linear soundness and $d_{\rm ss}=\Omega(n^\beta)$. Our theorem does not require explicit metachecks.}
\label{tab:single-shot-comparison}
\end{table*}
\section{Statistical-mechanical mapping}
\label{sec:smmapping}

In this section, we show how to express the logical error rate { for a chosen data decoder $c$ and maximum-likelihood final recovery} in terms of the thermal state of a statistical mechanical model. This facilitates physical analysis and numerical simulation~\cite{dennisTopologicalQuantumMemory2002,chubbStatisticalMechanical2021}.
The observed record of the preparation measurement is retained and may be used in the final recovery. We first describe the preparation process and the quantities it produces (Fig.~\ref{fig:syndrome-flow}), then define the logical error rate of maximum-likelihood recovery, and finally express this error rate through partition functions that can be evaluated by Monte Carlo sampling.

We define the auxiliary functions:

\begin{equation}
Q(s):=q^{|s|}(1-q)^{r_Z-|s|}\ \ , \ \ P(E):=p^{|E|}(1-p)^{n-|E|}
\end{equation}
In this section, the readout and data errors are independent Bernoulli errors with rates $q$ and $p$, so that $Q$ and $P$ are their probability distributions.

Figure~\ref{fig:syndrome-flow} in Sec.~\ref{sec:preliminaries} summarizes the single-shot state preparation procedure defined above. The system starts in $\ket{+}^{\otimes n}$ and undergoes a $Z$ syndrome measurement. The quantum state collapses to $X_{c(\sigma_0)}\ket{+_L}$, where the ideal first outcome $\sigma_0\in\im H_Z$ occurs with probability $2^{-r_Z'}$ and $r_Z'=\rank H_Z$. However, a readout error $s_e\in C^2$ affects the $Z$ syndrome readout. The observed record is therefore $s=\sigma_0+s_e\in C^2$, which need not be a valid syndrome.

The syndrome decoder $r$ of Eq.~\eqref{eq:syndrome-decoder} returns an estimate $r(s)\in C^2$ of the readout error such that $s+r(s)\in\im H_Z$, and it takes the same value on each coset of $\im H_Z$. Since $s$ and $s_e$ lie in the same coset, its output is
\begin{equation}
\label{eq:mapping-coset-record}
t:=r(s)=r(s_e)\in\mathcal R:=\im r.
\end{equation}
Thus $t$ is the representative of the coset $s_e+\im H_Z$ chosen by $r$, and $\mathcal R$ is the set of these representatives. We call $t$ the coset record. The feedback $X_{c(s+t)}$ applies the data decoder to the valid syndrome $s+t$ and leaves the residual syndrome
\begin{equation}
\label{eq:mapping-residual-syndrome}
\sigma:=\sigma_0+s+t=s_e+t\in\im H_Z
\end{equation}
of Eq.~\eqref{eq:residual-syndrome}. Since $X_{c(s+t)}X_{c(\sigma_0)}$ and $X_{c(\sigma)}$ have the same syndrome, they differ by an element of $Z^1$, which acts trivially on $\ket{+_L}$, and the quantum state becomes $X_{c(\sigma)}\ket{+_L}$. The coset record is therefore the common part of the two decompositions $s=(s+t)+t$ and $s_e=\sigma+t$ into a valid syndrome and an element of $\mathcal R$. Every readout error has a unique decomposition of this form, and the observed record determines $t$ but not $\sigma$.

Next, the system undergoes a Pauli $X$ data error $X_e$ ($e\in C^1$). The total error is $f:=c(\sigma)+e$. To compute the logical error rate, we then perform the ideal final recovery defined above. Its ideal $Z$ syndrome measurement gives the final syndrome $h:=H_Zf=\sigma+H_Ze$.

As in Eq.~\eqref{eq:preparation-channel}, we use $c(\sigma)$ as the effective preparation error. Since logical $X$ operators act trivially on $\ket{+_L}$, the prepared state alone does not fix this representative; the choice specifies the extension of the preparation noise to arbitrary encoded inputs, and $c(\sigma)$ is the representative analyzed in Theorem~\ref{thm:single-shot-preparation}. A record-dependent representative such as $c(\sigma_0)+c(s+t)$ differs from $c(\sigma)$ by $c(\sigma_0)+c(s+t)+c(\sigma)\in Z^1$. For a nonlinear data decoder, this difference can be a nontrivial logical operator (Appendix~\ref{app:mapping}), so the two representatives define different channels on general logical inputs. We therefore use $c(\sigma)$ throughout, for the actual total error as well as for the logical sectors below.

To estimate the logical error rate, we must infer the unknown readout error $s_e$ and data error $e$. The final recovery can use only the observed record $s$ and the final syndrome $h$. Using $\mathbf 1$ for the indicator function, uniformity of $\sigma_0$ gives
\begin{equation}
\label{eq:mapping-joint-law}
\Pr(s,s_e,e)=2^{-r_Z'}Q(s_e)P(e)\,\mathbf 1_{\{s+s_e\in\im H_Z\}}.
\end{equation}
Consequently, all observed records in the same coset, i.e., with the same coset record $t$, give the same posterior for $s_e$ and $e$, also after conditioning on $h$. Replacing $s$ by $t$ therefore loses no information, and from now on the records refer to the pair $(t,h)$.

Given the records $(t,h)$, the unknown quantities are the residual syndrome $\sigma$ and the data error $e$. They determine the readout error $s_e=\sigma+t$ and must reproduce the final syndrome $h=\sigma+H_Ze$. The joint probability of the records and of $(\sigma,e)$ is the probability $Q(\sigma+t)P(e)$ of these two errors, so the unnormalized conditional probability of $\sigma$ and $e$ is
\begin{equation}
\label{eq:mapping-posterior}
\Pr(\sigma,e\mid t,h)\propto
\begin{cases}
    Q(\sigma+t)P(e),&
    \substack{\sigma\in\im H_Z,\\ h=\sigma+H_Ze},\\
    0,&\text{otherwise}.
\end{cases}
\end{equation}

The final recovery applies a correction $X_g$ with $g\in C^1$ and $H_Zg=h$. The remaining error $X_{f+g}=X_{c(\sigma)+e+g}$ has zero syndrome, so the recovery succeeds if it is equivalent to the identity up to stabilizers and fails if it belongs to a different logical sector. We are not concerned with the exact values of $\sigma$ and $e$, nor with whether these values are inferred correctly. Therefore, maximum-likelihood recovery must select the most probable logical sector.

We define the logical sector. For each final syndrome $h$, we define a reference correction $c_0(h)\in C^1$ with $H_Zc_0(h)=h$. Any correction compatible with $h$ differs from this reference by stabilizers and logical $X$ operators:
\begin{equation}
\begin{gathered}
    g=c_0(h)+b+l,\qquad b\in B^1,\quad [l]\in H^1,\\
    \text{For each logical class }[l]\text{, choose a fixed}\\
    \text{representative vector }l\in Z^1.
\end{gathered}
\end{equation}
Here, $b\in B^1$ is a stabilizer, and $l$ is either a logical $X$ operator or the identity. The logical sector of the remaining error is $[f+g]=[f+c_0(h)]+[l]\in H^1$. A correction in the class $[l]$ succeeds precisely when $[f+c_0(h)]=[l]$, and we use $[l]$ to label this logical sector.

Let $\Pr(t,h,[l])$ denote the joint probability of observing the records $(t,h)$ and having an actual total error that satisfies $[f+c_0(h)]=[l]$. For records with $\Pr(t,h)>0$, the posterior probability of the logical sector $[l]$ is
\begin{equation}
\label{eq:mapping-sector-posterior}
P_{t,h}([l]):=\Pr\left([l]\mid t,h\right)=\frac{\Pr(t,h,[l])}{\Pr(t,h)}.
\end{equation}
To minimize the logical error rate, maximum-likelihood recovery selects $\argmax_{[l]\in H^1}P_{t,h}([l])$ and fails with the conditional probability
\begin{equation}
P_{\rm fail}^{\rm opt}(t,h):=1-\max_{[l]\in H^1}P_{t,h}([l]).
\end{equation}
The average success probability is therefore
\begin{equation}
\label{eq:mapping-pmld}
\begin{aligned}
    P_{\rm MLD}
    &=\sum_{t,h}\Pr(t,h)\max_{[l]}P_{t,h}([l])\\
    &=\sum_{t,h}\max_{[l]}\Pr(t,h,[l]),
\end{aligned}
\end{equation}
and the logical error rate of maximum-likelihood recovery is $1-P_{\rm MLD}$.
In all record sums, $t\in\mathcal R$ and $h\in\im H_Z$; records with $\Pr(t,h)=0$ never occur and are omitted from all averages of conditional quantities. Changing $c_0(h)$ only permutes the logical labels for each pair of records, so $P_{\rm MLD}$ is independent of this reference choice. The data decoder $c$ defining the preparation channel remains unchanged. The ideal final recovery used in Theorem~\ref{thm:single-shot-preparation} applies $X_{c(h)}$. Taking $c_0=c$, its success probability satisfies
\begin{equation}
\label{eq:mapping-fixed-recovery}
    1-P_{\rm fail}(p,q)=\sum_{t,h}\Pr(t,h,[0])\leq P_{\rm MLD}.
\end{equation}
Thus the threshold guarantee for the fixed final recovery also gives a threshold for maximum-likelihood recovery of the same effective channel.

We now express these probabilities as partition functions. For a candidate residual syndrome $\sigma$, the condition $[f+c_0(h)]=[l]$ is equivalent to $e=c(\sigma)+c_0(h)+b+l$ for one $b\in B^1$, and every such $e$ satisfies $h=\sigma+H_Ze$. Summing the weights in Eq.~\eqref{eq:mapping-posterior} over these configurations gives
\begin{equation}
\label{eq:mapping-sector-sum}
\begin{aligned}
    \mathcal Z(t,h,[l])
    &:=\Pr(t,h,[l])\\
    &=\sum_{\sigma\in\im H_Z}\sum_{b\in B^1}
    Q(\sigma+t)\\
    &\quad\times P\!\left(c(\sigma)+c_0(h)+b+l\right),
\end{aligned}
\end{equation}
which we interpret as the partition function within the logical sector $[l]$. The total probability serves as the partition function:
\begin{equation}
\label{eq:mapping-joint-normalization}
    \mathcal Z(t,h)
    :=\sum_{[l]\in H^1}\mathcal Z(t,h,[l])
    =\Pr(t,h),
\end{equation}
with $\sum_{t\in\mathcal R}\sum_{h\in\im H_Z}\mathcal Z(t,h)=1$. Hence $P_{t,h}([l])=\mathcal Z(t,h,[l])/\mathcal Z(t,h)$ is the probability of the sector $[l]$ in the thermal distribution associated with $\mathcal Z(t,h)$, and maximum-likelihood recovery identifies the most probable logical sector of this distribution. The records $(t,h)$ play the role of quenched disorder with distribution $\mathcal Z(t,h)$. For a function $A$ of the records, we write $[A]_{\rm dis}:=\sum_{t,h}\mathcal Z(t,h)A(t,h)$ for its disorder average, so that $1-P_{\rm MLD}=[P_{\rm fail}^{\rm opt}]_{\rm dis}$.

We have thus recast the problem as a problem in statistical mechanics. To sample it, we combine the sums over $\sigma$ and $b$ into a single sum over configurations $x\in C^1$. A configuration $x$ describes the candidate data error $e=x+c_0(h)$, so that $H_Zx=\sigma$ and the logical sector $[f+c_0(h)]$ becomes $[\ell_c(x)]$, where $\ell_c(x)$ is defined below. The target thermal distribution is
\begin{equation}
\label{eq:mapping-thermal}
\begin{aligned}
\Pr(x\mid t,h)
&=\frac{Q(t+H_Zx)P\!\left(x+c_0(h)\right)}{\mathcal Z(t,h)},\\
\ell_c(x)&:=x+c(H_Zx),
\end{aligned}
\end{equation}
and $\mathcal Z(t,h,[l])$ is the total weight of the configurations with $[\ell_c(x)]=[l]$. Each sampled configuration is assigned the logical sector $[\ell_c(x)]$ using the same data decoder $c$ that defines the preparation error (Appendix~\ref{app:mapping}).

For $0<p,q<1$, the Bernoulli weights also give an explicit classical spin model. Define $\xi_i=(-1)^{x_i}$, $K_p=\frac12\log\frac{1-p}{p}$, and $K_q=\frac12\log\frac{1-q}{q}$. Absorbing the inverse temperature into the couplings, the Hamiltonian is
\begin{equation}
\label{eq:mapping-hamiltonian}
\begin{aligned}
\mathcal H_{t,h}(\xi)={}&-K_q\sum_{a=1}^{r_Z}(-1)^{t_a}\prod_{i:(H_Z)_{ai}=1}\xi_i\\
&-K_p\sum_{i=1}^{n}(-1)^{c_0(h)_i}\xi_i,
\end{aligned}
\end{equation}
and the weight in Eq.~\eqref{eq:mapping-thermal} is proportional to $\exp[-\mathcal H_{t,h}(\xi)]$. Each interaction in the first term is supported on a measured $Z$ check, and the second term is a local field. The logical sectors are specified by $[\ell_c(x)]$; no additivity of $c$ is used in defining these sectors.

One way to perform this sampling is through MCMC, by repeatedly updating $x$ and collecting samples. Finally, we compute the logical error rate of interest using samples collected after the Markov chain reaches stationarity.
\par

\subsection{Order parameter}
\label{subsec:mapping-order-parameter}

The maximum-likelihood decoder chooses the most probable logical class.
We now quantify how strongly the records $(t,h)$ favor a single class.
Choose dual logical bases identifying $H_1$ and $H^1$ with $\mathbb F_2^k$.
Their natural binary pairing assigns a character
$\chi_u([l])=(-1)^{u\cdot[l]}$ to each $u\in H_1$. We define

\begin{equation}
\label{eq:mapping-logical-moments}
m_u(t,h)
    :=\sum_{[l]\in H^1}\chi_u([l])P_{t,h}([l]),
\end{equation}

and, for $k\geq1$,

\begin{equation}
\label{eq:mapping-qtop}
q_{\rm top}(t,h)
    :=\frac{1}{2^k-1}\sum_{u\ne0}m_u(t,h)^2.
\end{equation}

This quantity measures how strongly the posterior is concentrated on
logical classes. Its disorder average weights each pair of records
by $\mathcal Z(t,h)$:

\begin{equation}
\label{eq:mapping-qtop-average}
{}[q_{\rm top}]_{\rm dis}
    =\sum_{t,h}
    \mathcal Z(t,h)q_{\rm top}(t,h).
\end{equation}

Equivalently, we independently draw readout and data error realizations $s_e$ and $e$ from $Q$ and $P$, and set
\begin{equation}
\label{eq:mapping-disorder-draw}
t=r(s_e),\qquad \sigma=s_e+t,\qquad h=\sigma+H_Ze.
\end{equation}
The induced distribution of the records $(t,h)$ is $\mathcal Z(t,h)$, giving the exact identity
\begin{equation}
\label{eq:mapping-sampling-identity}
\begin{aligned}
{}[q_{\rm top}]_{\rm dis}
={}&\sum_{s_e\in C^2}\sum_{e\in C^1}Q(s_e)P(e)\\
&\times q_{\rm top}\bigl(r(s_e),s_e+r(s_e)+H_Ze\bigr).
\end{aligned}
\end{equation}
Thus each disorder realization and the thermal distribution conditioned on it use the same noise model and the same effective preparation error.

\begin{proposition}
\label{prop:qtop-mld} For the fixed data decoder $c$, with the observed record retained,

\begin{equation}
\label{eq:mapping-qtop-bound}
\frac{1+(2^k-1)[q_{\rm top}]_{\rm dis}}{2^k}
    \leq P_{\rm MLD}
    \leq\sqrt{\frac{1+(2^k-1)[q_{\rm top}]_{\rm dis}}{2^k}}.
\end{equation}

Consequently, along any family with $k\geq1$,

\begin{equation}
[q_{\rm top}]_{\rm dis}\longrightarrow1
    \quad\Longleftrightarrow\quad
    P_{\rm MLD}\longrightarrow1.
\end{equation}

\end{proposition}

The proof uses the Parseval identity for the logical characters
and is given in Appendix~\ref{app:mapping}. For fixed $k$, a uniform logical
posterior has $q_{\rm top}=0$, whereas a posterior concentrated on one
class has $q_{\rm top}=1$. Intermediate values can describe partial logical
information. The proposition characterizes the limit of successful block
decoding; it does not imply that every unsuccessful regime has
$q_{\rm top}\to0$.

To evaluate this criterion by Monte Carlo sampling, we express the logical
expectation values as thermal averages of microscopic observables.
Recall that a configuration $x$ lies in the logical sector
$[\ell_c(x)]$ of Eq.~\eqref{eq:mapping-thermal}.
Choose a cycle $\bar z_u\in Z_1$ representing $u\in H_1$, and define

\begin{equation}
\label{eq:mapping-microscopic-observable}
O_u^{(c)}(x) :=(-1)^{\langle\bar z_u,\ell_c(x)\rangle}.
\end{equation}

Since $H_Z\ell_c(x)=0$, this observable depends only on the
cohomology class of $\ell_c(x)$. Its thermal average under the target
distribution $\Pr(x\mid t,h)$ defined above satisfies

\begin{equation}
\label{eq:mapping-microscopic-identity}
\bigl\langle O_u^{(c)}(x)\bigr\rangle_{t,h}
    =m_u(t,h).
\end{equation}

We use this explicitly defined logical observable as the generalized
Wilson loop in the mapping. It includes the correction inferred from the
syndrome, so evaluating it can require information outside the support of
$\bar z_u$.
Equations~\eqref{eq:mapping-qtop-average}
and~\eqref{eq:mapping-microscopic-identity} therefore provide a Monte Carlo
prescription for the maximum-likelihood decoding threshold. We can estimate
the squared thermal averages using two independent equilibrium replicas
at the same fixed records $(t,h)$ and noise
parameters $(p,q)$.
The Parseval identity also gives an equivalent estimator based on the
probability that the two replicas have the same logical class.

\subsection{Numerical study}

\label{sec:numerics}

\subsubsection{Code family}

We consider the 3D toric code, which provides a test
case with check redundancy and nonlinear soundness. We use the toric
convention with qubits on edges of a periodic cubic lattice, $X$ checks on
the six edges incident on a vertex, and $Z$ checks on the four edges
surrounding a plaquette~\cite{castelnovoTopologicalOrder2008}. Products of
plaquette checks satisfy local cube relations and global relations on
noncontractible planes. These relations constrain the valid syndromes
and must be included in the syndrome decoder. In the dual lattice, an $X$
error is a surface whose boundary is the plaquette syndrome. Small
contractible loops can bound surfaces with area quadratic in their length.
The resulting relation between syndrome weight and minimum correction
weight is nonlinear, so this geometry lies outside the linear-soundness
{ hypothesis in part~1 of Theorem~\ref{thm:single-shot-preparation}.}

According to Definition~\ref{def:soundness}, the toric-code
$Z$-syndrome map detecting $X$ errors is $(L-1,f)$-sound with
$L-1=\Theta(n^{1/3})$ and $f(x)=Cx^2$, where $C$ is independent of $L$.
For $|\sigma|<L$, the dual syndrome loops admit a filling of area at most
$C|\sigma|^2$ (Appendix~\ref{app:code-families}). Thus the minimum
correction weight $w$ obeys $w\leq C|\sigma|^2$, or equivalently
$|\sigma|\geq\sqrt{w/C}$; this square-root dependence is the superlinear
soundness referred to below. Square membrane boundaries attain
$w=\Theta(|\sigma|^2)$, ruling out an $L$-independent linear soundness bound. {
These minimum-weight membranes can also be chosen with
growing area within any polynomially growing correction-weight range.
They therefore violate the linear confinement condition in part~2. We study the 3D toric code to explore whether the single-shot
state-preparation threshold persists
in the superlinear-soundness regime beyond Theorem~\ref{thm:single-shot-preparation}.}

\begin{figure*}[t]
\centering
\newcommand{\toricpanel}[2]{%
\begin{minipage}[t]{0.40\textwidth}
\textbf{(#1)}\\
\IfFileExists{#2.pdf}%
{\includegraphics[width=\linewidth]{#2.pdf}}%
{\fbox{\parbox[c][0.68\linewidth][c]{0.95\linewidth}{\centering
\textcolor{red}{Placeholder: \texttt{\detokenize{#2.pdf}}}}}}
\end{minipage}}%
\toricpanel{a}{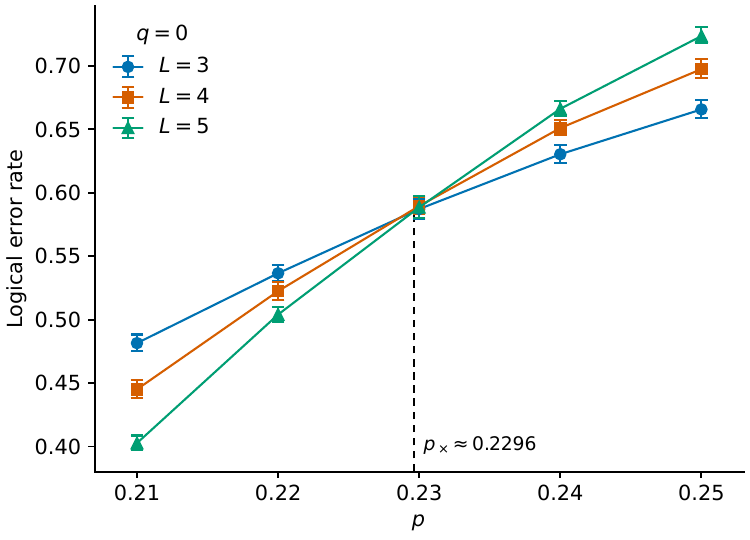}\hspace{0.05\textwidth}
\toricpanel{b}{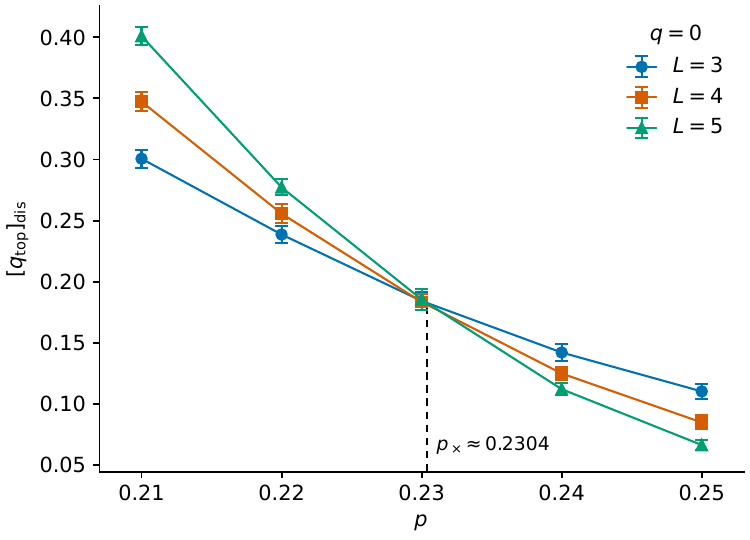}\\[1ex]
\toricpanel{c}{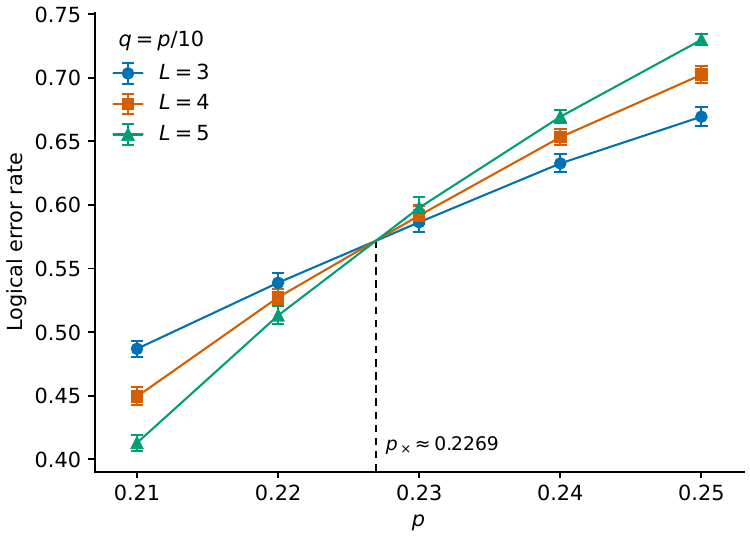}\hspace{0.05\textwidth}
\toricpanel{d}{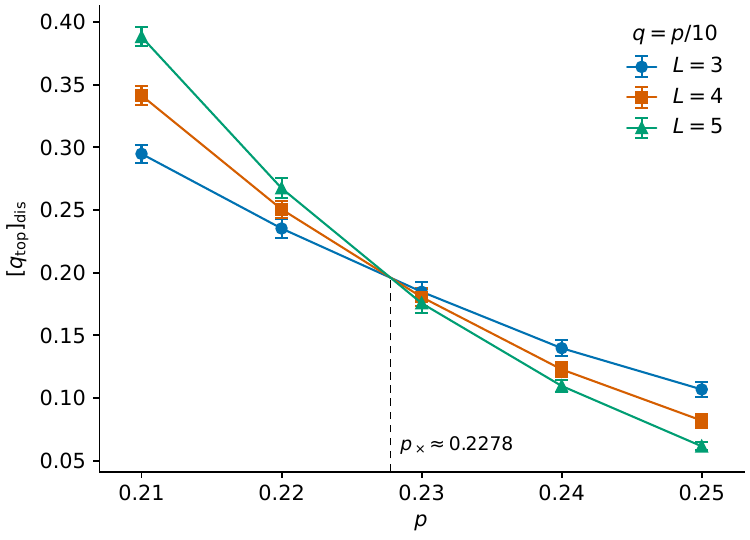}\\[1ex]
\toricpanel{e}{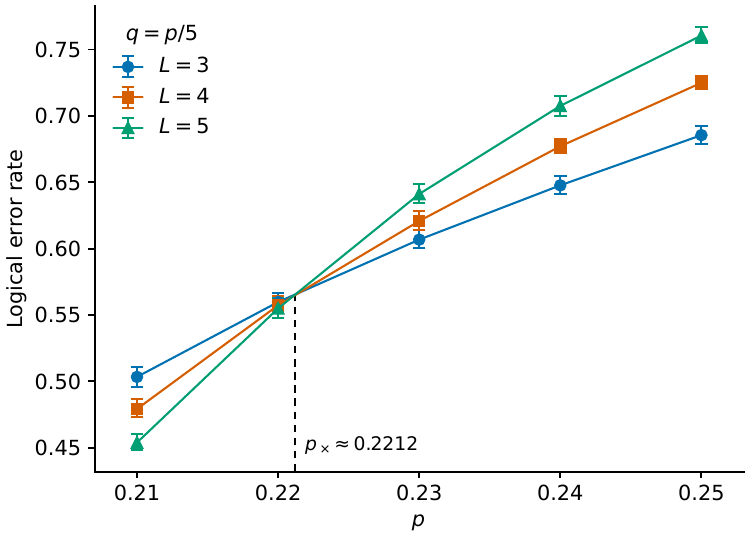}\hspace{0.05\textwidth}
\toricpanel{f}{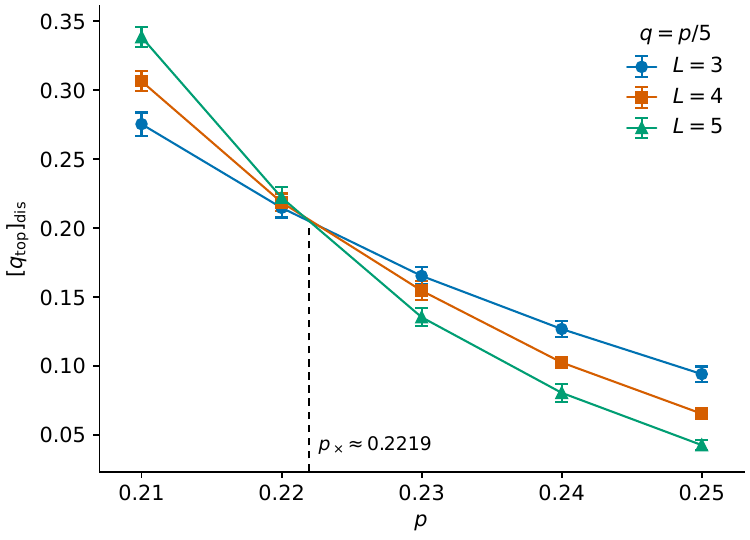}
\caption{\textbf{Numerical evidence for a nonzero single-shot
state-preparation threshold in the three-dimensional toric code
with readout errors.}
Panels (a,b) show noiseless readout, $q=0$, while panels (c,d)
and (e,f) include readout errors with $q=p/10$ and $q=p/5$, respectively,
where $p$ and $q$ are the data and readout error probabilities.
At lower noise, increasing the system size suppresses the logical
error rate after ideal maximum-likelihood final recovery (left)
and increases the logical-sector order parameter
$[q_{\rm top}]_{\rm dis}$ (right); both size dependences reverse
at higher noise.
The persistence of these crossings at nonzero readout noise
supports a single-shot preparation threshold even with superlinear
soundness.
Periodic cubic lattices have $L=3,4,5$, $n=3L^3$, and $k=3$.
Vertical dashed lines mark approximate finite-size crossing
positions. 
Data points show the estimates and error bars the confidence
intervals over independent disorder realizations described in
Sec.~\ref{sec:numerics}.}
\label{fig:toric-numerics}
\end{figure*}

\subsubsection{Sampling procedure}

For each disorder realization, we independently draw Bernoulli readout
and data errors $s_e$ and $e$ with probabilities $q$ and $p$, and set
$t$, $\sigma$, and $h$ as in Eq.~\eqref{eq:mapping-disorder-draw}.
The records $(t,h)$ then have the joint probability $\mathcal Z(t,h)$
of Eq.~\eqref{eq:mapping-joint-normalization}, and give the
disorder average in Eq.~\eqref{eq:mapping-qtop-average}.
The ideal first outcome $\sigma_0$ does not enter the records $(t,h)$
or the total error $f$, and is therefore not sampled.

For the syndrome decoder $r$, we use the minimum-weight decoder of
Eq.~\eqref{eq:syndrome-decoder}, with ties broken by a fixed
lexicographic order (Appendix~\ref{app:code-families}). For the data
decoder $c$, we use belief propagation with
localized statistics decoding (BPLSD)~\cite{hillmann2025lsd,roffe2022ldpc}
as a computationally accessible substitute for exact minimum-weight
decoding, with $H_Zc(\sigma)=\sigma$ verified for each correction. The
same fixed function $c$ generates the effective preparation error
$c(\sigma)$ of each disorder realization and labels the thermal
configurations. Since this $c$ is nonlinear, the choice of $r$ is also
part of the protocol (Appendix~\ref{app:mapping}). We choose a fixed
linear right inverse $c_0(h)=Gh$ of $H_Z$ as the reference correction;
this choice only permutes logical-class labels and leaves both
observables unchanged.

At fixed records $(t,h)$, we sample $x\in C^1$ from the thermal
distribution proportional to
$Q(t+H_Zx)P(x+c_0(h))$.
Each sampled configuration $x$ is assigned the logical class
$[\ell_c(x)]$ defined in
Eq.~\eqref{eq:mapping-thermal}, using the same fixed
data decoder $c$ throughout. At $q=0$, the readout error vanishes, so
that $t=0$ and
$\sigma=0$, the thermal distribution is supported on $x\in Z^1$, and
$\ell_c(x)=x$ because $c(0)=0$.

For each disorder realization, we run two independent Markov chains
(replicas) at the same records, and denote by $\hat P^{(1)}_{t,h}$ and
$\hat P^{(2)}_{t,h}$ their normalized logical-class frequencies. We
estimate the logical error rate $[1-\max_{[l]}P_{t,h}([l])]_{\rm dis}$
by $[1-\max_{[l]}\bar P_{t,h}([l])]_{\rm dis}$, where
$\bar P_{t,h}=(\hat P^{(1)}_{t,h}+\hat P^{(2)}_{t,h})/2$ are the
pooled frequencies. We also compute the cross-validated estimate
\begin{equation}
\label{eq:numerical-cv}
\Bigl[1-\tfrac12\hat P^{(2)}_{t,h}\bigl(\hat l^{(1)}\bigr)
-\tfrac12\hat P^{(1)}_{t,h}\bigl(\hat l^{(2)}\bigr)\Bigr]_{\rm dis},
\end{equation}
where $\hat l^{(r)}$ is the most frequent class of replica $r$, with
ties broken by a fixed order of the class labels. In this estimate, the
class is selected with one replica and its probability is evaluated
with the other. For independent equilibrated replicas, the pooled and
cross-validated estimates bracket the logical error rate from below and
above on average. For the order parameter, we use
\begin{equation}
\label{eq:numerical-normalization}
q_{\rm top}(t,h)
 =\frac{2^k\sum_{[l]}P_{t,h}([l])^2-1}{2^k-1},
\end{equation}
which follows from the Parseval identity, and estimate
$\sum_{[l]}P_{t,h}([l])^2$ by the overlap
$\sum_{[l]}\hat P^{(1)}_{t,h}([l])\hat P^{(2)}_{t,h}([l])$ of the two
replicas. As a consistency check, the Nishimori property in
Appendix~\ref{app:mapping} expresses $[q_{\rm top}]_{\rm dis}$ through
the posterior probability of the true logical class $[f+c_0(h)]$.
Evaluated with the pooled frequency of this class, it agrees with the
overlap estimate within statistical errors.

Means and standard errors are taken over independent disorder
realizations. Error bars are confidence intervals constructed from
these standard errors, with critical values that account for the
repeated evaluations made while the disorder realizations were
accumulated. For the logical error rate, the lower end of each interval
is based on the pooled estimate and the upper end on the
cross-validated estimate. Convergence of $[q_{\rm top}]_{\rm dis}$ to
one corresponds to asymptotically successful maximum-likelihood
decoding by Proposition~\ref{prop:qtop-mld}.

\subsubsection{Numerical results}

Our MCMC simulations sample the thermal distribution by combining
$X$-stabilizer moves, single-qubit moves for $q>0$, and logical moves.
The stabilizer moves add to $x$ the support of one $X$ check or of the
product of two adjacent $X$ checks, and leave both $H_Zx$ and the
logical class unchanged; they and the single-qubit moves are accepted
with the Metropolis rule, i.e., with probability equal to the smaller of
one and the ratio of the new to the old thermal weight. The logical move chooses, for each of the
three logical directions, a logical $X$ representative supported on a
noncontractible membrane at a uniformly random position. It then
resamples $x$ among the eight configurations obtained by adding
products of these three representatives, with probabilities
proportional to their thermal weights. At $q=0$, the single-qubit
moves are disabled, and the stabilizer and logical moves connect all
of $Z^1$. For part of the parameter points, including all points with
$L=5$, we additionally use parallel
tempering~\cite{hukushimaExchangeMonteCarlo1996} with four or five
inverse temperatures $\beta$ between $0.55$ and $1$. The chain at
inverse temperature $\beta$ samples a distribution proportional to the
$\beta$-th power of the thermal weight, and observables are measured
only at $\beta=1$. We estimate the observables from equilibrated
samples, checked by comparing the first and second halves of each
chain and the two replicas, and average over disorder realizations.
The simulation parameters are listed in
Appendix~\ref{app:code-families}.

Fig.~\ref{fig:toric-numerics} shows a reversal of the size dependence
in both observables: at smaller $p$, increasing $L$ reduces the logical
error rate and increases $[q_{\rm top}]_{\rm dis}$, whereas the ordering
reverses at larger $p$. This crossing behavior persists at nonzero
readout noise, $q=p/10$ and $q=p/5$, supporting a nonzero preparation
threshold in the superlinear-soundness setting. Along the lines $q=0$,
$q=p/10$, and $q=p/5$, the logical-error-rate curves cross near
$p\approx0.230$, $0.227$, and $0.221$, respectively, and the crossings
of $[q_{\rm top}]_{\rm dis}$ lie within $0.001$ of these values.
{ This trend suggests
that superlinear soundness can support a preparation threshold. }

\section{Discussion and Outlook}
\label{sec:conclusions}

Our results distinguish two levels of protection in single-shot state preparation. Under linear soundness, the residual preparation error is local stochastic, so standard modular fault-tolerance arguments apply. Linear confinement does not guarantee this property, but logical protection can still be established with a consistent data decoder: for neighboring syndromes, the assigned corrections differ by the intervening single-qubit error up to a stabilizer. Our confinement-based threshold therefore does not require a local-stochastic bound on the residual error. This makes the data decoder part of the structure governing the robustness of single-shot preparation, rather than merely a post-processing step.

The confinement-based theorem has a limitation. Its threshold is certified by an ideal final recovery, a standard device in analyses of single-shot fault tolerance~\cite{quintavalleSingleShotErrorCorrection2021}, and does not yet establish composition with a noisy single-shot correction cycle. The key question is whether decoder consistency can be preserved when syndrome extraction and recovery are noisy. If so, confinement-based preparation could fit into a modular fault-tolerance framework even without a local-stochastic bound on the intermediate residual error; if not, preparation and subsequent correction may need to be analyzed jointly.

Decoder dependence also points to a broader role for decoding in state preparation. The soundness result uses minimum-weight decoding, while the confinement result proves the existence of a consistent decoder without giving a general efficient construction; the BPLSD decoder used numerically is not known to satisfy this property. The statistical-mechanical mapping likewise optimizes the final recovery while keeping the preparation decoder fixed. These results make decoder design part of the preparation problem itself. An important next step is to construct efficient consistent decoders and to determine whether recoverability can be optimized or characterized without fixing a particular preparation decoder.

The theorem gives sufficient structural conditions, but not the full boundary of recoverability. In the retained-record formulation, the statistical-mechanical mapping probes logical-sector order directly, and the three-dimensional toric-code results suggest recoverability beyond both linear conditions. Although the present data use maximum-likelihood final recovery and sizes only up to $L=5$, they point to a broader connection between code geometry, decoder consistency, and statistical-mechanical order. Establishing this connection could lead to more general criteria for single-shot initialization beyond soundness-based constructions.

\section{Acknowledgements}
This work was supported by Quantum Science and Technology--National Science and Technology Major Project (Grant No.~2021ZD0302400), National Natural Science Foundation of China (Grant No.~92365111), and Shanghai Municipal Science and Technology (Grant No.~25LZ2600200).

\appendix
\section{Proof of the preparation threshold with linear soundness range}
\label{app:threshold}

We prove Theorem~\ref{thm:single-shot-preparation} with local stochastic
noise as defined in Definition~\ref{def:local-stochastic}.
The syndrome decoder $r$ assigns a fixed minimum-weight representative
to each coset of $\im H_Z$. Consequently,
$r(s+\sigma)=r(s)$ for every valid syndrome $\sigma$.

\subsection{Minimum-weight decoding and connected components}

We first show that a residual valid syndrome requires readout errors
on at least half of its support.

\begin{lemma}
\label{lemma:half-weight}
Let $\pi(s)=s+r(s)$. If $b\subseteq\pi(s)$ and
$b\in\im H_Z$, then $|s\cap b|\geq |b|/2$.
\end{lemma}
\begin{proof}
The vectors $r(s)$ and $r(s)+b$ belong to the same syndrome coset.
Minimum weight therefore gives
\begin{equation}
    |r(s)|\leq |r(s)+b|
    =|r(s)|+|b|-2|r(s)\cap b|.
\end{equation}
On $b\subseteq s+r(s)$, the supports of $s$ and $r(s)$ are
complementary. It follows that
\begin{equation}
    |s\cap b|=|b|-|r(s)\cap b|\geq |b|/2.
\end{equation}
\end{proof}

To apply this constraint to residual data errors, we decompose their
supports into connected components. Define the graph $G(H_Z)$ with one
vertex for each qubit and an edge between two qubits that share a $Z$
check. Its maximum degree is at most $\Delta=(d_c-1)d_v$.
We identify binary errors with their supports and write
$w\sqsubseteq c(\sigma)$ when $w$ is a union of connected components
of the support of $c(\sigma)$ on this graph.

\begin{lemma}
\label{lemma:UCC}
If $w\sqsubseteq c(\sigma)$, then $H_Zw\subseteq\sigma$ and
$|c(H_Zw)|=|w|$.
\end{lemma}
\begin{proof}
Write $c(\sigma)=w\sqcup f$. Since $w$ and $f$ are disconnected,
no check meets both supports. Their syndromes are therefore disjoint,
so $H_Zw\subseteq\sigma$ follows from
$\sigma=H_Zw+H_Zf$. To prove the second claim, suppose some $w'$
satisfies $H_Zw'=H_Zw$ and $|w'|<|w|$. Then $w'+f$ has syndrome
$\sigma$ and weight at most $|w'|+|f|<|c(\sigma)|$,
contradicting the definition of $c$. Thus $w$ is a minimum-weight
representative of its syndrome. The argument does not require it to
coincide with a particular decoder output.
\end{proof}

We next bound the probability that an error contains a prescribed qubit
set. Let $\mathcal W(e)$ be the family of qubit sets $w$ that contain $e$
and whose connected components each intersect $e$. Write
$M_\Delta(\omega,e)$ for the number of members with $|w|=\omega$.
For nonempty $e$, Lemma~2 of
Ref.~\cite{gottesmanFaultTolerantQuantumComputation2014} gives the bound
\begin{equation}
\label{eq:component-count}
    M_\Delta(\omega,e)
    \leq \e^{|e|-1}(\Delta\e)^{\omega-|e|}.
\end{equation}
The sets counted here may be disconnected, provided each connected
component intersects $e$.

\subsection{Local stochasticity of the preparation error}

\begin{lemma}
\label{lemma:local-stochastic-syndrome}
Suppose $H_Z$ is $(\kappa n,\rho)$-linearly sound and has bounded
check and qubit degrees. Let $S_e$ be $q$-local stochastic and
$\Sigma=\pi(S_e)$. Define
\begin{equation}
    \alpha=\min\{\rho,\kappa\},\qquad
    \widetilde q=\e(2\sqrt q)^\alpha.
\end{equation}
If $\Delta\widetilde q\leq 1/2$, then $c(\Sigma)$ is
$\widetilde q$-local stochastic.
\end{lemma}
\begin{proof}
If a fixed nonempty set $e$ is contained in $c(\Sigma)$, the union of
components of $c(\Sigma)$ that meet $e$ belongs to $\mathcal W(e)$. Hence
\begin{equation}
\label{eq:component-event-bound}
    \Pr(e\subseteq c(\Sigma))
    \leq\sum_{w\in\mathcal W(e)}
    \Pr(w\sqsubseteq c(\Sigma)).
\end{equation}
By Lemma~\ref{lemma:UCC}, each event in the sum is empty unless
$w$ is a minimum-weight representative of $H_Zw$. For such a $w$,
linear soundness gives $|H_Zw|\geq\rho|w|$ whenever
$|H_Zw|\leq\kappa n$. If instead $|H_Zw|>\kappa n$, then
$|H_Zw|\geq\kappa|w|$ because $|w|\leq n$. Thus, in either case,
\begin{equation}
\label{eq:component-soundness}
    |H_Zw|\geq\alpha|w|.
\end{equation}

On the event $w\sqsubseteq c(\Sigma)$, we have
$H_Zw\subseteq\Sigma$. Lemma~\ref{lemma:half-weight} therefore implies
$|S_e\cap H_Zw|\geq |H_Zw|/2$. Applying the local stochastic
assumption and a union bound gives
\begin{equation}
\begin{aligned}
    \Pr(w\sqsubseteq c(\Sigma))
    &\leq
    \sum_{j=\lceil |H_Zw|/2\rceil}^{|H_Zw|}
    \binom{|H_Zw|}{j}q^j\\
    &\leq (2\sqrt q)^{|H_Zw|}
    \leq (2\sqrt q)^{\alpha|w|}.
\end{aligned}
\end{equation}
The small-noise condition ensures $q<1/4$, as required for the last
inequality. The bound also holds for sets $w$ that do not have minimum
weight for their syndrome, since the corresponding event has probability
zero. Thus the soundness estimate is used only for sets that can occur
as unions of connected components of minimum-weight corrections.

Combining Eqs.~\eqref{eq:component-count} and
\eqref{eq:component-event-bound} gives
\begin{equation}
\begin{aligned}
    \Pr(e\subseteq c(\Sigma))
    &\leq\sum_{\omega=|e|}^{n}
    \e^{|e|-1}(\Delta\e)^{\omega-|e|}
    (2\sqrt q)^{\alpha\omega}\\
    &\leq
    \frac{\widetilde q^{|e|}}
    {\e(1-\Delta\widetilde q)}
    \leq\widetilde q^{|e|}.
\end{aligned}
\end{equation}
The last inequality uses $1/[\e(1-\Delta\widetilde q)]\leq2/\e<1$.
For $e=\emptyset$, the desired bound reduces to $1\leq1$.
The case $q=0$ follows directly from $S_e=0$ and $c(0)=0$.
\end{proof}

We now include the independent data noise. If $E_1$ and $E_2$ are
independent local stochastic errors with parameters $p_1$ and $p_2$,
their sum is local stochastic with parameter $p_1+p_2$.
To see this, every qubit in a support contained in $E_1+E_2$ must
belong to at least one of the two errors. Hence
\begin{equation}
\begin{aligned}
    \Pr(e\subseteq E_1+E_2)
    &\leq\sum_{e'\subseteq e}
    \Pr(e'\subseteq E_1)
    \Pr(e\setminus e'\subseteq E_2)\\
    &\leq\sum_{e'\subseteq e}
    p_1^{|e'|}p_2^{|e|-|e'|}
    =(p_1+p_2)^{|e|}.
\end{aligned}
\end{equation}
The first line uses independence. Applying this bound to
$F=E+c(\Sigma)$ gives the noise parameter
$s(p,q)=p+\widetilde q(q)$ stated in the theorem.

\subsection{Logical failure after ideal recovery}

Consider a perfect syndrome measurement followed by minimum-weight
recovery $R=c(H_ZF)$, and write $D=F+R$. Since $H_ZD=0$,
every connected component $W$ of $D$ has zero syndrome, because
distinct components meet disjoint sets of checks. If $D$ represents
a nontrivial logical $X$ operator, at least one component must also
represent such an operator and therefore satisfy $|W|\geq d_X$.

The recovery $R+W$ has the same syndrome as $R$. The minimum-weight
property of the chosen recovery therefore implies
\begin{equation}
    |R|\leq|R+W|,
    \qquad |R\cap W|\leq |W|/2.
\end{equation}
The supports of $F$ and $R$ are complementary on $W\subseteq D$,
so $|F\cap W|\geq |W|/2$. For a fixed connected set $W$ of size
$\omega$, local stochasticity gives
\begin{equation}
    \Pr(|F\cap W|\geq\omega/2)
    \leq\sum_{j=\lceil\omega/2\rceil}^{\omega}
    \binom{\omega}{j}s^j
    \leq(2\sqrt s)^\omega.
\end{equation}
Applying Eq.~\eqref{eq:component-count} with one prescribed vertex and
summing over the $n$ choices of that vertex bounds the number of connected
sets of size $\omega$ by
$n(\e\Delta)^{\omega-1}$. Summing over these sets gives, for
$2\e\Delta\sqrt s<1$,
\begin{equation}
\begin{aligned}
    P_{\rm fail}(p,q)
    &\leq\sum_{\omega=d_X}^n
    n(\e\Delta)^{\omega-1}(2\sqrt s)^\omega\\
    &\leq\frac{n}{\e\Delta}
    \frac{(2\e\Delta\sqrt s)^{d_X}}
    {1-2\e\Delta\sqrt s}.
\end{aligned}
\end{equation}

The bound in Eq.~\eqref{eq:main-failure-bound} gives the following
sufficient thresholds for the combined noise strength:
\begin{equation}
\label{eq:distance-dependent-threshold}
    s_{\rm th}=
    \begin{cases}
        \dfrac{\exp(-2/c_d)}{4\e^2\Delta^2},
        & d_X\geq c_d\log n,\\[6pt]
        \dfrac{1}{4\e^2\Delta^2},
        & d_X=\omega(\log n).
    \end{cases}
\end{equation}
Here $c_d>0$ is independent of $n$, and $\log$ denotes the natural
logarithm. The stronger distance assumption permits a larger threshold.
For every fixed $p,q$ with $s(p,q)<s_{\rm th}$, we show that
$P_{\rm fail}(p,q)\to0$ as $n\to\infty$.

These thresholds also ensure the conditions used to bound the preparation
error. Indeed, $\widetilde q(q)\leq s(p,q)<s_{\rm th}\leq
1/(4\e^2\Delta^2)$ and $\Delta>1$ imply $q<1/4$ and
$\Delta\widetilde q(q)<1/2$. The case $s(p,q)=0$ is immediate.
Otherwise, let $a=2\e\Delta\sqrt{s(p,q)}$, so that $0<a<1$.

For the first case, $d_X\geq c_d\log n$ implies
$n\leq\exp(d_X/c_d)$ for all sufficiently large $n$. Hence
\begin{equation}
\label{eq:logarithmic-distance-decay}
    P_{\rm fail}(p,q)
    \leq\frac{\exp\!\left[-d_X\left(\log(1/a)-1/c_d\right)\right]}
    {\e\Delta(1-a)}.
\end{equation}
The threshold gives $a<\exp(-1/c_d)$, so the coefficient
$\log(1/a)-1/c_d$ is positive. The logical failure probability therefore
decays exponentially with $d_X$ and tends to zero as $n\to\infty$.

For the second case, $d_X/\log n\to\infty$ and $a<1$ imply
\begin{equation}
\label{eq:superlogarithmic-distance-decay}
    \log\!\left(n a^{d_X}\right)
    =\log n-d_X\log(1/a)\longrightarrow-\infty.
\end{equation}
Thus $P_{\rm fail}(p,q)\to0$. Since $\log n=o(d_X)$, the same bound
also gives exponential decay with $d_X$ in this case.

The constants in Theorem~\ref{thm:single-shot-preparation} can therefore
be chosen as
\begin{equation}
\label{eq:separate-noise-thresholds}
    p_0=\frac{s_{\rm th}}{2},\qquad
    q_0=\frac14\left(\frac{s_{\rm th}}{2\e}\right)^{2/\alpha}.
\end{equation}
For $p<p_0$ and $q<q_0$, we have
$s(p,q)=p+\widetilde q(q)<s_{\rm th}$, which proves the threshold
statement in both cases.

For every error configuration on which recovery succeeds, $F+R$ is an
$X$ stabilizer and acts trivially on
the entire code space, including encoded states entangled with a
reference system. The same probability bound therefore controls the
entanglement infidelity of the channel
$\mathcal N\circ\mathcal S$ after recovery.

{
\section{Proof of the threshold under linear confinement}
\label{app:polynomial-threshold}

Here we prove part~2 of Theorem~\ref{thm:single-shot-preparation}. Throughout this section, we assume
$(t,\rho)$-linear confinement in the sense of
Definition~\ref{def:confinement}, with $t=\Omega(\log n)$.
We adopt the noise model and qubit graph $G(H_Z)$ of
Appendix~\ref{app:threshold}, but change the syndrome and data decoders.
The proof first bounds errors within connected sets, then constructs
decoders that preserve logical information under such errors. The
cluster counting bound in Eq.~\eqref{eq:component-count} gives the final
probability estimate.

\subsection{Confinement within connected sets}

In addition to $G(H_Z)$, let $G_s$ be the graph whose vertices are the
measured $Z$ checks, with two checks adjacent when they share a qubit.
Its maximum degree is at most $d_c(d_v-1)$. Identity checks may be
omitted, so $r_Z\leq d_vn$. For a vertex set $A$ in either graph, define the closeness weight \cite{quintavalleSingleShotErrorCorrection2021}
\begin{equation}
\label{eq:polynomial-closeness}
 \|A\|_\ell=
 \max_{\substack{K\text{ connected}\\1\leq |K|\leq\ell}}|A\cap K|.
\end{equation}
It allows connected
sets of size at most $\ell$ to include graphs with small connected
components. The maximum is zero for an empty graph. We use $G(H_Z)$
for data errors and $G_s$ for syndrome errors. The quantity is monotone
under inclusion and satisfies the triangle inequality for binary addition.

Throughout the proof, take
\begin{equation}
\label{eq:polynomial-proof-scale}
 \ell=\left\lfloor\min\left\{
 t,\frac{d_X}{2}\right\}\right\rfloor.
\end{equation}
The assumptions of the theorem ensure $\ell\geq8$ for all sufficiently
large $n$. 

A key property of the closeness weight is that it preserves confinement, as in Lemma 5 of Ref. \cite{quintavalleSingleShotErrorCorrection2021}. However, here we consider the reduced weight that minimizes $\|e+b\|_\ell$ over $b\in B^1$ when comparing
errors that differ by an $X$ stabilizer. This minimization preserves the
logical class and differs from the minimization over all errors with
the same syndrome in Ref.~\cite{quintavalleSingleShotErrorCorrection2021}.
Adding a single-qubit error changes this minimum by at most one.
\begin{lemma}
\label{lemma:polynomial-confinement}
Suppose $H_Z$ has $(t,\rho)$-linear confinement.
If $\min_{b\in B^1}\|e+b\|_\ell\leq\ell/2$, then
\begin{equation}
\label{eq:polynomial-confinement}
 \rho\min_{b\in B^1}\|e+b\|_\ell
 \leq\|H_Ze\|_{2d_v\ell}.
\end{equation}
\end{lemma}
\begin{proof}
Choose $f=e+b$ with $b\in B^1$ to minimize $\|f\|_\ell$.
If $\|f\|_\ell=0$, the claim holds. Otherwise, every connected
component $f_i$ of $f$ has weight at most $\|f\|_\ell\leq\ell/2$.
Indeed, a larger component would contain a connected set of
$\|f\|_\ell+1\leq\ell$ qubits, contradicting the definition of
$\|f\|_\ell$.

For each $i$, choose a minimum-weight correction $g_i$ with
$H_Zg_i=H_Zf_i$. Since $|g_i|\leq|f_i|$, the zero-syndrome error
$f_i+g_i$ has weight at most $2|f_i|\leq\ell<d_X$ and is therefore
an $X$ stabilizer. Thus $g=\sum_i g_i$ differs from $e$ by an
$X$ stabilizer. The choice of $f$ therefore implies
$\|g\|_\ell\geq\|f\|_\ell$.

Although $f_i$ is connected, its minimum-weight correction $g_i$ need not be connected. By Lemma~\ref{lemma:UCC}, every connected component $u$ of each $g_i$
is a minimum-weight representative of its syndrome. Within each
$g_i$, these component syndromes are disjoint subsets of $H_Zg_i$.
Moreover, the syndromes $H_Zg_i=H_Zf_i$ are disjoint for different
$i$, since the components $f_i$ are disconnected. Thus all the
component syndromes are pairwise disjoint subsets of $H_Ze$. Since
\begin{equation}
 |u|\leq\ell/2\leq t,
\end{equation}
each $u$ is a minimum-weight correction within the confinement range.
Definition~\ref{def:confinement} therefore gives
\begin{equation}
 |H_Zu|\geq\rho|u|.
\end{equation}

Choose a connected qubit set $K$ with $|K|\leq\ell$ and
$|g\cap K|\geq\|f\|_\ell$. Every qubit in $g\cap K$ belongs to at
least one component $u$ that intersects $K$, so the sum of the weights
of these components is at least $\|f\|_\ell$.
Select components one at a time until the sum of their weights is
at least $\|f\|_\ell$. Since each component has weight at most
$\ell/2$, this sum is at most $\|f\|_\ell+\ell/2\leq\ell$.
Together with $K$, the selected components therefore form a connected
set of at most $2\ell$ qubits. The checks incident on this set form
a connected set of at most $2d_v\ell$ vertices in $G_s$ and include
all the selected component syndromes. Since these syndromes are
disjoint subsets of $H_Ze$, the check set contains at least
$\sum_u|H_Zu|\geq\rho\sum_u|u|\geq\rho\|f\|_\ell$ vertices of
$H_Ze$, where both sums run over the selected components.
This proves Eq.~\eqref{eq:polynomial-confinement}.
\end{proof}

\subsection{Syndrome and data decoders}

We construct the data decoder by successively adding single-qubit
errors to a reference correction.
We first define the following graph of valid syndromes:
\begin{equation}
\label{eq:polynomial-valid-syndromes}
 \mathcal V=\left\{\sigma\in\im H_Z:
 \|\sigma\|_{2d_v\ell}\leq\frac{\rho\ell}{8}\right\}.
\end{equation}
Join each pair of syndrome configurations in $\mathcal V$ by an edge labelled $j$ when their
difference is $H_Ze_j$, where $e_j$ is the single-qubit error on
qubit $j$. We distinguish parallel edges by their qubit labels.

\begin{lemma}
\label{lemma:polynomial-paths}
For every closed path in this graph, the sum of the single-qubit
errors specified by its edge labels is an $X$ stabilizer.
\end{lemma}
\begin{proof}
Let $\sigma_0,\ldots,\sigma_m$ be a path, and let $u_j$ be the sum
of the single-qubit errors specified by its first $j$ edges, with
$u_0=0$. Then
$H_Zu_j=\sigma_j+\sigma_0$, so
\begin{equation}
 \|H_Zu_j\|_{2d_v\ell}
 \leq\|\sigma_j\|_{2d_v\ell}+\|\sigma_0\|_{2d_v\ell}
 \leq\frac{\rho\ell}{4}.
\end{equation}
We claim that $\min_{b\in B^1}\|u_j+b\|_\ell\leq\ell/4$ at every
step along the path. The bound holds for $u_0=0$, and adding a
single-qubit error changes the minimum by at most one. At the first
step violating the bound, the minimum would therefore be at most
$\ell/4+1\leq\ell/2$. However,
Lemma~\ref{lemma:polynomial-confinement} would then imply that the
minimum is at most $\ell/4$, a contradiction.

If the path is closed, then $H_Zu_m=0$. As shown in
Appendix~\ref{app:threshold}, every representative of a nontrivial
logical $X$ operator contains a connected component of weight at least
$d_X$. Such a component contains a connected $\ell$-qubit set, so
every representative has closeness weight $\ell$. This contradicts
the bound just proved if $u_m$ represents a nontrivial logical
operator. Hence $u_m\in B^1$.
\end{proof}

The key point of the decoder is that its correction operators are chosen according to the correlations between syndrome configurations induced by single-qubit steps.
To define $c$, choose a base syndrome in each connected component
of the graph on $\mathcal V$ and a data error with that syndrome
as its correction.
For any other syndrome in the same component, add to this data error
the single-qubit errors along a path from the base syndrome.
By Lemma~\ref{lemma:polynomial-paths}, different paths give errors
that differ by an $X$ stabilizer. Choose one representative of this
class as $c(\sigma)$. In the component containing zero, choose the
zero syndrome and zero error, with $c(0)=0$. For each valid syndrome
outside $\mathcal V$, choose any data error with that syndrome.
Thus $H_Zc(\sigma)=\sigma$ for all valid syndromes, and the
corrections at the endpoints of every edge in the graph satisfy
\begin{equation}
\label{eq:polynomial-decoder-consistency}
 c(\sigma)+e_j+c(\sigma+H_Ze_j)\in B^1.
\end{equation}

For syndrome decoding, choose one representative of each coset of
$\im H_Z$ that minimizes the closeness weight:
\begin{equation}
\label{eq:polynomial-syndrome-decoder}
 r(s)\in\argmin_{s'\in C^2:\,s+s'\in\im H_Z}
 \|s'\|_{2d_v\ell}.
\end{equation}
The choice depends only on the coset, so $r(s+\sigma)=r(s)$ for every
valid syndrome $\sigma$. Equation~\eqref{eq:residual-syndrome}
therefore gives $\Sigma=S_e+r(S_e)$. Since $S_e$ belongs to the
same coset as $r(S_e)$, we have
$\|r(S_e)\|_{2d_v\ell}\leq\|S_e\|_{2d_v\ell}$.
The triangle inequality then gives
\begin{equation}
\label{eq:polynomial-residual-syndrome}
 \|\Sigma\|_{2d_v\ell}\leq2\|S_e\|_{2d_v\ell}.
\end{equation}
The final recovery uses the same data decoder $c$. Both decoders
are independent of the noise distribution.

\begin{lemma}
\label{lemma:polynomial-success}
The preparation channel, followed by data noise and ideal recovery,
preserves the logical information whenever
\begin{equation}
\label{eq:polynomial-good-noise}
 \|S_e\|_{2d_v\ell}\leq\frac{\rho\ell}{32},\qquad
 \|E\|_{2d_cd_v\ell}\leq\frac{\rho\ell}{16d_v}.
\end{equation}
\end{lemma}
\begin{proof}
For any subset $E'\subseteq E$, the syndrome satisfies
\begin{equation}
\label{eq:polynomial-data-syndrome}
 \|H_ZE'\|_{2d_v\ell}
 \leq d_v\|E\|_{2d_cd_v\ell}.
\end{equation}
Indeed, the qubits incident on a connected set of at most $2d_v\ell$
checks form a connected set of at most $2d_cd_v\ell$ qubits.
Each qubit in $E'$ contributes to at most $d_v$ of these checks.
Combining this bound with Eqs.~\eqref{eq:polynomial-residual-syndrome}
and~\eqref{eq:polynomial-good-noise} gives
\begin{equation}
 \|\Sigma+H_ZE'\|_{2d_v\ell}
 \leq2\|S_e\|_{2d_v\ell}+d_v\|E\|_{2d_cd_v\ell}
 \leq\frac{\rho\ell}{8}.
\end{equation}
Starting from $\Sigma$, we can therefore apply the single-qubit
errors in $E$ in any order to obtain a path entirely within
$\mathcal V$. Summing
Eq.~\eqref{eq:polynomial-decoder-consistency} along the path yields
\begin{equation}
 c(\Sigma)+E+c(\Sigma+H_ZE)\in B^1.
\end{equation}
Since $F=E+c(\Sigma)$ and $H_ZF=\Sigma+H_ZE$, the left-hand side
is the error remaining after recovery. It is an $X$ stabilizer and
therefore acts trivially on the code space. As in
Appendix~\ref{app:threshold}, recovery also preserves encoded states
entangled with a reference system.
\end{proof}

\subsection{Logical failure probability}

By Lemma~\ref{lemma:polynomial-success}, a logical error can occur
only if at least one of the bounds in
Eq.~\eqref{eq:polynomial-good-noise} is violated. To bound this
probability, consider a graph with $N$ vertices and maximum degree
at most $D>1$. Equation~\eqref{eq:component-count} bounds the number of connected
$j$-vertex sets by $N(\e D)^{j-1}$. For a $p$-local stochastic
error $E$, the probability that such a set contains more than
$h>0$ errors is at most $2^j p^h$. A union bound over the connected
sets therefore gives, for
$0<p<1$ and a positive integer $m$,
\begin{equation}
\label{eq:polynomial-probability}
\begin{aligned}
 \Pr(\|E\|_m>h)
 &\leq\sum_{j=1}^m N(\e D)^{j-1}2^j p^h\\
 &\leq N(2\e D)^m p^h.
\end{aligned}
\end{equation}
The case $p=0$ follows from $E=0$ almost surely.

Recall the definition of $a(p,q)$ in part~2 of
Theorem~\ref{thm:single-shot-preparation}:
\begin{equation}
\label{eq:polynomial-noise-bound}
\begin{aligned}
 a(p,q)=\max\bigl\{& (2\e\Delta)^{2d_cd_v}p^{\rho/(16d_v)},\\
 & [2\e\max\{2,d_c(d_v-1)\}]^{2d_v}q^{\rho/32}\bigr\}.
\end{aligned}
\end{equation}
Apply Eq.~\eqref{eq:polynomial-probability} to the data errors on
$G(H_Z)$ and the readout errors on $G_s$. Using $r_Z\leq d_vn$
and taking a union bound for the two violations of
Eq.~\eqref{eq:polynomial-good-noise} gives
\begin{equation}
\label{eq:polynomial-explicit-bound}
 P_{\rm fail}(p,q)\leq(1+d_v)n\,a(p,q)^\ell.
\end{equation}
We now apply the argument of Appendix~\ref{app:threshold}
to the growth of both the confinement range and the distance.
Since $t,d_X=\Omega(\log n)$, choose a constant $c_d>0$ such that
$d_X\geq c_d\log n$ and $t\geq(c_d/2)\log n$ for all sufficiently
large $n$.
Equation~\eqref{eq:polynomial-proof-scale} implies
$\ell\geq(c_d/2)\log n-1$. Hence, for $0<a(p,q)<\exp(-2/c_d)$,
\begin{equation}
\begin{aligned}
 P_{\rm fail}(p,q)
 \leq{}&(1+d_v)\exp(2/c_d)\\
 &\times\exp\!\left[-\ell\left(\log\frac{1}{a(p,q)}
                                  -\frac{2}{c_d}\right)\right].
\end{aligned}
\end{equation}
As in Eq.~\eqref{eq:logarithmic-distance-decay}, this bound decays
exponentially with $\ell$. Since $a(p,q)\to0$ as $p,q\to0$,
we can choose constants $p_0,q_0>0$ independent of $n$ so that
$a(p,q)<\exp(-2/c_d)$ for all $p<p_0$ and $q<q_0$. If $a(p,q)=0$,
the failure probability is zero by
Eq.~\eqref{eq:polynomial-explicit-bound}.

If $t=\omega(\log n)$ and $d_X=\omega(\log n)$, then $\ell=\omega(\log n)$. In this case,
the argument in Eq.~\eqref{eq:superlogarithmic-distance-decay}
gives exponential decay with $\ell$ whenever $a(p,q)<1$.
The stronger growth assumptions therefore give larger sufficient
thresholds for the readout and data error rates.
In both cases, $\ell=\Theta(\min\{t,d_X\})$, so the failure
probability is exponentially small in $\min\{t,d_X\}$ below
nonzero readout and data error thresholds. This proves
Eq.~\eqref{eq:polynomial-failure-bound} and completes the proof of
part~2 of Theorem~\ref{thm:single-shot-preparation}.
}

\section{An explicit qLDPC family with linear soundness}
\label{app:ramanujan}

We describe the family used in the main text and verify the hypotheses
of part~1 of Theorem~\ref{thm:single-shot-preparation}. The code is
defined by the incidence relations among the vertices, edges, and
triangles of a simplicial complex.

\subsection{Simplicial complexes and the CSS code}

A simplicial complex is a collection of finite vertex sets, called
faces, that contains every subset of each face. A face with $i+1$
vertices has dimension $i$. Thus vertices, edges, triangles, and
tetrahedra have dimensions zero through three. For a complex $X_j$,
we write $X_j(i)$ for its set of $i$-dimensional faces. Its
two-skeleton $Y_j=X_j^{(2)}$ consists of its vertices, edges, and
triangles, with their incidence relations unchanged.

All chain and cochain spaces have coefficients in $\Ftwo$.
A chain in $C_i(Y_j)$ is
a binary vector indexed by $Y_j(i)$, or equivalently a selection of
$i$-dimensional faces. A cochain in $C^i(Y_j)$ assigns a binary value
to each such face. In these bases, both spaces are identified with
$\Ftwo^{Y_j(i)}$. The boundary map $\partial_1$ sends each edge
to the sum of its two endpoints, and $\partial_2$ sends each triangle
to the sum of its three edges. All sums are modulo two, so
$\partial_1\partial_2=0$. The coboundary maps are the transposes,
$\delta_0=\partial_1^T$ and $\delta_1=\partial_2^T$.

Place one qubit on each edge and define
\begin{equation}
    H_Z=\partial_2^T=\delta_1,
    \qquad H_X=\partial_1.
\end{equation}
Each triangle specifies a $Z$ check on its three edges, and each
vertex specifies an $X$ check on its incident edges. The boundary
identity gives $H_XH_Z^T=0$, as required for a CSS code. For an
edge error $a$, the syndrome bit $(\delta_1a)(\tau)$ is the parity
of the erroneous edges in triangle $\tau$. We denote the numbers
of qubits and $Z$ checks by $n_j=|Y_j(1)|$ and $m_j=|Y_j(2)|$.

In the notation of the main text, the one-cocycles
$Z^1(Y_j)=\ker\delta_1$ are the edge sets that meet every triangle
in an even number of edges. The one-coboundaries
$B^1(Y_j)=\im\delta_0$ are the supports of products of vertex
$X$ checks. Hence $H^1(Y_j;\Ftwo)=Z^1(Y_j)/B^1(Y_j)$ describes
the logical $X$ operators: a cocycle is nontrivial if it is not
a coboundary. Similarly, the one-cycles
$Z_1(Y_j)=\ker\partial_1$ are the edge sets with even incidence
at every vertex, and the one-boundaries
$B_1(Y_j)=\im\partial_2$ are sums of triangle boundaries.
The quotient $H_1(Y_j;\Ftwo)=Z_1(Y_j)/B_1(Y_j)$ describes the
logical $Z$ operators. These spaces depend only on the two-skeleton,
since adding tetrahedra does not change any of the maps just defined.

\subsection{Choice of the complexes and code parameters}

Choose a sufficiently large odd prime power $Q$, independent of $j$,
and let $\mathbb F_Q$ be the field with $Q$ elements.
We first describe the infinite complex from which the finite
complexes are obtained. Let $\mathbb F_Q((t))$ be the field of formal
Laurent series with coefficients in $\mathbb F_Q$ and only finitely
many negative powers of $t$, and let $\mathcal O=\mathbb F_Q[[t]]$
be its subring of power series with no negative powers. A lattice
in $\mathbb F_Q((t))^4$ is the $\mathcal O$-span of a basis of this
four-dimensional vector space. Identify lattices that differ by
multiplication by a nonzero scalar. These equivalence classes are
the vertices of the Bruhat--Tits building $\mathcal B$.
For $1\leq i\leq3$, classes $[L_0],\ldots,[L_i]$ form an
$i$-dimensional face when they have representatives satisfying
\begin{equation}
    tL_0\subsetneq L_i\subsetneq\cdots
    \subsetneq L_1\subsetneq L_0.
\end{equation}
Since $L_0/tL_0$ is four-dimensional over $\mathbb F_Q$, maximal
such chains have four lattice classes. Thus $\mathcal B$ is
three-dimensional~\cite{lubotzkySamuelsVishneExplicit2005}.

A quotient of $\mathcal B$ identifies vertices and faces in the
same orbit under a group of permutations of its vertices that
preserve faces. The
constructions in Ref.~\cite{lubotzkySamuelsVishneExplicit2005}
and Ref.~\cite{evraKaufmanZemorDecodable2022} give an
explicit sequence of finite simplicial quotients $X_j$ with
\begin{equation}
    |X_j(0)|\longrightarrow\infty,
    \qquad H^1(X_j;\Ftwo)\neq0.
\end{equation}
For these quotients, the projection from $\mathcal B$ preserves
the local incidence relations: the faces containing a vertex of
$X_j$ correspond bijectively to those containing any chosen
preimage of that vertex. In $\mathcal B$, the neighbors of a vertex
correspond to the nonzero proper subspaces of $\mathbb F_Q^4$,
and the faces through that vertex correspond to nested chains of
these subspaces. Consequently the numbers of incident edges,
triangles, and tetrahedra are bounded by constants depending only
on $Q$. The associated CSS codes therefore have bounded check
weights and qubit degrees, independently of $j$.

The one-dimensional cosystole is the minimum number of edges in
a nontrivial one-cocycle; the one-dimensional systole is the
corresponding minimum for a nontrivial one-cycle. The cosystolic
and systolic bounds in Ref.~\cite{evraKaufmanZemorDecodable2022},
applied to this family,
give a linear lower bound in $|X_j(0)|$ for the cosystole and a
logarithmic lower bound for the systole. Since the number of edges
incident on each vertex is bounded above and below by positive
constants independent of $j$, $n_j=\Theta(|X_j(0)|)$. Thus
\begin{equation}
\begin{aligned}
    k_j&=\dim H^1(Y_j;\Ftwo)>0,\\
    d_{X,j}&=\min_{a\in Z^1(Y_j)\setminus B^1(Y_j)}|a|
              =\Omega(n_j),\\
    d_{Z,j}&=\min_{a\in Z_1(Y_j)\setminus B_1(Y_j)}|a|
              =\Omega(\log n_j).
\end{aligned}
\end{equation}
The constants in these lower bounds may depend on $Q$ but not
on $j$. The distance labels follow the convention $H_Z=\delta_1$
used throughout this paper.

\subsection{From cosystolic expansion to linear soundness}

To state the expansion bound, we first define the weighted norms
on the two-skeleton. For an edge cochain $u$ and a triangle
cochain $v$, let
\begin{equation}
\label{eq:LSV-weighted-norms}
\begin{aligned}
    \|u\|_1
    &=\frac{1}{3m_j}
      \sum_{e\in\supp(u)}
      \bigl|\{\tau\in Y_j(2):e\subset\tau\}\bigr|,\\
    \|v\|_2&=\frac{|v|}{m_j}.
\end{aligned}
\end{equation}
The first norm is the probability of selecting an edge in $\supp(u)$
when choosing a triangle uniformly and then one of its three edges
uniformly. The second is the fraction of triangles in $\supp(v)$.
The subscripts specify the cochain degree.

In degree one, cosystolic expansion requires a uniform positive
lower bound on the weighted norm of each nontrivial cocycle, together
with a constant $\epsilon>0$ such that
\begin{equation}
\label{eq:LSV-global-cosystolic}
    \|\delta_1a\|_2
    \geq\epsilon\min_{z\in Z^1(Y_j)}\|a+z\|_1
\end{equation}
for every $a\notin Z^1(Y_j)$. The minimum is the weighted distance
from $a$ to the space of cocycles. Equivalently, it is the minimum
weighted norm of an error with syndrome $\delta_1a$.
The cosystolic expansion result and its extension to finite building
quotients in Ref.~\cite{evraKaufmanBoundedDegree2017} establish
these two properties for $Y_j$ when $Q$ is sufficiently large,
with constants independent of $j$. We use
Eq.~\eqref{eq:LSV-global-cosystolic} to prove soundness.

Let $r_-$ and $r_+$ be positive lower and upper bounds on the number
of triangles containing an edge. The local incidence relations
described above give $1\leq r_-\leq r_+<\infty$, independently
of $j$. Counting edge--triangle incidences gives
\begin{equation}
    r_-n_j\leq3m_j\leq r_+n_j,
\end{equation}
and the weighted edge norm satisfies
\begin{equation}
    \|u\|_1\geq\frac{r_-}{3m_j}|u|.
\end{equation}
For a nonzero valid syndrome $\sigma=\delta_1a$, all errors with
syndrome $\sigma$ form the set $a+Z^1(Y_j)$. Therefore the
minimum-weight data decoder $c_{\rm min}$ satisfies
\begin{equation}
    |c_{\rm min}(\sigma)|
    =\min_{z\in Z^1(Y_j)}|a+z|.
\end{equation}
Combining this equality with Eq.~\eqref{eq:LSV-global-cosystolic}
and the norm bound gives
\begin{equation}
\begin{aligned}
    \frac{|\sigma|}{m_j}
    &\geq\epsilon\min_{z\in Z^1(Y_j)}\|a+z\|_1\\
    &\geq\frac{\epsilon r_-}{3m_j}
           \min_{z\in Z^1(Y_j)}|a+z|\\
    &=\frac{\epsilon r_-}{3m_j}|c_{\rm min}(\sigma)|.
\end{aligned}
\end{equation}
Since $c_{\rm min}(0)=0$, the same bound holds for the zero
syndrome. Hence
\begin{equation}
    |c_{\rm min}(\sigma)|\leq\frac{|\sigma|}{\rho},
    \qquad \rho=\frac{\epsilon r_-}{3}>0
\end{equation}
for every valid syndrome. In particular, $H_Z$ is
$(\kappa n_j,\rho)$-linearly sound with $\kappa=r_-/3$, where
both $\kappa$ and $\rho$ are independent of $j$.
Each $Z$ check has weight $d_c=3$, each qubit belongs to at most
$r_+$ such checks, and $\Delta=2r_+$ bounds the degree of the
qubit adjacency graph. Together with the nonzero code dimension and
the distance bounds above, these estimates verify the hypotheses
of part~1 of Theorem~\ref{thm:single-shot-preparation}.

\section{Partition functions and logical-sector averages}
\label{app:mapping}

We first justify the procedure used in the main text to sample disorder
from the noise distribution. The map $r$ is constant on each coset of
$\im H_Z$, and $r(t)=t$ for $t\in\mathcal R$. Each coset contains
$2^{r_Z'}$ elements of the form $t+\sigma$, with $\sigma\in\im H_Z$.
Sampling $s_e$ from $Q$ and setting $\sigma=\pi(s_e)$ therefore assigns
weight $Q(\sigma+t)$ to the residual syndrome $\sigma$ and the selected
coset $t$. When the complete preparation record is retained,
Eq.~\eqref{eq:mapping-joint-law} assigns weight
$2^{-r_Z'}Q(\sigma+t)$ to each of the $2^{r_Z'}$ measurement records $s$
in that coset. The two procedures thus give the same average for any
observable that is constant on each coset of $\im H_Z$. Since the
total error $f=e+c(\sigma)$ and the final syndrome $h$ do not depend on
the first ideal outcome, the partition functions, logical expectation
values, and conditional success probabilities used above all have this
property. In particular, $\Pr(s,h,[l])=2^{-r_Z'}\mathcal Z(r(s),h,[l])$,
and Eq.~\eqref{eq:mapping-sampling-identity} gives the disorder average
directly in terms of the readout and data noise distributions. These
identities hold for nonlinear $c$ as well as linear $c$.

We combine the logical-sector sums into a single sum over qubit
variables. To obtain the unrestricted partition function, set
$\sigma=H_Zx$ and decompose the cocycle $x+c(\sigma)=b+l$. This defines
a bijection between $x\in C^1$ and triples $(\sigma,b,[l])$ in the
partition function of Eq.~\eqref{eq:mapping-joint-normalization}.
Under this bijection, the weight $Q(t+H_Zx)P\bigl(x+c_0(h)\bigr)$
coincides with the summand of $\Pr(t,h,[l])$ in
Eq.~\eqref{eq:mapping-sector-sum}. The expression for the logical observable
follows from $O_u^{(c)}(x)=(-1)^{u\cdot[l]}$ for the corresponding
triple, which gives Eq.~\eqref{eq:mapping-microscopic-identity}.
Changing the representative $\bar z_u$ by an element of $B_1$ leaves
this observable unchanged, since $\ell_c(x)$ is a cocycle.

The disorder of Eq.~\eqref{eq:mapping-disorder-draw} and the thermal
distribution are related by the Nishimori
property~\cite{nishimoriInternalEnergy1981,nishimoriStatisticalPhysics2001},
which holds because both use the same noise model. For the actual
realization, $x^*=e+c_0(h)$ satisfies $H_Zx^*=\sigma$ and
$s_e=t+H_Zx^*$, so the joint probability of $x^*$ and the records is
$Q(t+H_Zx^*)P(x^*+c_0(h))$. Conditional on $(t,h)$, $x^*$ is therefore
an exact sample of the thermal distribution, and its logical class
$[\ell_c(x^*)]=[f+c_0(h)]$ is the true sector. Consequently,
$[q_{\rm top}]_{\rm dis}=(2^k\,\mathbb E[P_{t,h}([f+c_0(h)])]-1)/(2^k-1)$,
where the expectation is over disorder realizations.

In contrast to $c_0$, the functions $r$ and $c$ define the effective
channel. If $c$ is additive modulo $B^1$, replacing $r$ by another
valid syndrome decoder changes the residual syndrome by a valid
syndrome determined by $t$, which can be absorbed into a
record-dependent reference correction; both observables are then
independent of $r$. For a nonlinear $c$, this argument fails, and the
representative in Sec.~\ref{sec:smmapping} also matters. As a minimal
example, let
$H_Z=\left(\begin{smallmatrix}1&1&0\\0&1&1\end{smallmatrix}\right)$
with no $X$ checks, so that $X_{111}$ is a logical operator. The
minimum-weight decoder $c(10)=100$, $c(01)=001$, $c(11)=010$ gives
$c(10)+c(01)+c(11)=111$. For $\sigma_0=10$ and $\pi(s)=01$, the
representatives $c(\sigma_0)+c(\pi(s))$ and $c(\sigma)$ therefore
differ by a logical operator.

\begin{proof}[Proof of Proposition~\ref{prop:qtop-mld}]
For fixed $t$ and $h$, the Parseval identity gives
\begin{equation}
\label{eq:mapping-parseval}
\begin{aligned}
    \sum_{[l]}P_{t,h}([l])^2
    &=2^{-k}\sum_{u\in H_1}m_u(t,h)^2\\
    &=\frac{1+(2^k-1)q_{\rm top}(t,h)}{2^k},
\end{aligned}
\end{equation}
where $m_0=1$. Every probability distribution satisfies
\begin{equation}
    \sum_{[l]}P_{t,h}([l])^2
    \leq\max_{[l]}P_{t,h}([l])
    \leq\sqrt{\sum_{[l]}P_{t,h}([l])^2}.
\end{equation}
The first inequality follows by bounding one factor in each squared
term by the largest probability. The second follows because the sum
contains the square of that largest probability. Averaging with
$\mathcal Z(t,h)$ gives the lower bound in
Eq.~\eqref{eq:mapping-qtop-bound}, and concavity of the square root
gives the upper bound. If $[q_{\rm top}]_{\rm dis}\to1$, the lower
bound implies $P_{\rm MLD}\to1$. Conversely, combining the upper
bound with $k\geq1$ gives
\begin{equation}
    0\leq1-[q_{\rm top}]_{\rm dis}
    \leq\frac{2^k}{2^k-1}(1-P_{\rm MLD}^2)
    \leq2(1-P_{\rm MLD}^2),
\end{equation}
which proves the converse even when $k$ grows with the block size.
\end{proof}

To derive the Fourier representation and its normalization, we insert
the binary Fourier representation
\begin{equation}
    \delta(\sigma+H_Zx)
    =2^{-r_Z}\sum_{\tau\in C^2}
    (-1)^{\langle\tau,\sigma+H_Zx\rangle}
\end{equation}
into
$\sum_{\sigma,x}Q(t+\sigma)P\bigl(x+c_0(h)\bigr)
\delta(\sigma+H_Zx)$. The independent binary sums satisfy
\begin{equation}
\begin{aligned}
    &\sum_{\sigma\in C^2}Q(t+\sigma)
    (-1)^{\langle\tau,\sigma\rangle}
    =(-1)^{\langle\tau,t\rangle}(1-2q)^{|\tau|},\\
    &\sum_{x\in C^1}P\bigl(x+c_0(h)\bigr)
    (-1)^{\langle H_Z^T\tau,x\rangle}
    \\&=(-1)^{\langle H_Z^T\tau,c_0(h)\rangle}
    (1-2p)^{|H_Z^T\tau|}.
\end{aligned}
\end{equation}
Using $\langle H_Z^T\tau,c_0(h)\rangle
=\langle\tau,h\rangle$ gives
\begin{equation}
\label{eq:mapping-fourier}
\begin{aligned}
 &\sum_{x\in C^1}Q(t+H_Zx)P\bigl(x+c_0(h)\bigr)\\
 &\quad=2^{-r_Z}\sum_{\tau\in C^2}
 (-1)^{\langle\tau,t+h\rangle}
 (1-2q)^{|\tau|}\\
 &\qquad\times(1-2p)^{|H_Z^T\tau|}.
\end{aligned}
\end{equation}
The left-hand side equals $\mathcal Z(t,h)$.
In particular, the record distribution does not depend on the data
decoder $c$; the decoder enters only through the logical-sector labels.

\section{Code constructions for the numerical comparison}
\label{app:code-families}

On the periodic cubic lattice, $H_X$ is the vertex--edge incidence
matrix and $H_Z$ is the plaquette--edge incidence matrix over $\Ftwo$.
Since each plaquette has an even number of edges incident on each
vertex, $H_XH_Z^T=0$. Each $Z$ check has weight four, and each edge
belongs to four plaquettes, giving $d_c=d_v=4$ for this syndrome map.
The $X$ checks have weight six. The code has three logical qubits:
logical $Z$ operators are noncontractible paths, and logical $X$
operators are noncontractible dual membranes~\cite{castelnovoTopologicalOrder2008}.

The six plaquettes bounding each cube form a check relation, and the
three independent noncontractible closed coordinate planes give additional
relations. A valid syndrome must satisfy all of them. In the dual
lattice, the cube constraints require even incidence at each vertex,
while the plane constraints remove the nontrivial homology of the
syndrome. A decoder that enforces only local cube parity can therefore
return a closed dual loop with nontrivial homology, which is not in
$\im H_Z$. The definition of $r$ in Eq.~\eqref{eq:syndrome-decoder}
enforces membership in the space of valid syndromes.

For the numerics, we collect $L^3-1$ independent cube relations and the
three plane relations into a matrix $M$ with $MH_Z=0$ and
$\rank M=r_Z-\rank H_Z=L^3+2$, so that $\ker M=\im H_Z$. For each record
$s$, $r(s)$ is the lexicographically smallest vector among the
minimum-weight solutions $a\in C^2$ of $Ma=Ms$, where vectors are
compared at their first differing check, in a fixed ordering of the
checks, with $0<1$. It depends on $s$ only
through its coset, satisfies $r(0)=0$, and is a minimum-weight decoder
in the sense of Eq.~\eqref{eq:syndrome-decoder}. We compute it exactly.
Solutions of weight at most three are found by exhaustive search.
Otherwise, minimum-weight perfect matching on the graph whose vertices
are cubes and whose edges are plaquettes enforces the cube relations
and gives a lower bound and a candidate; an integer program including
the plane relations gives the optimal weight when the candidate
violates them. The lexicographic representative is then fixed bit by
bit, with a bit set to one only when the alternative is proven
infeasible.

The data decoder $c$ is the BPLSD decoder of the \texttt{ldpc}
package~\cite{roffe2022ldpc} (version 2.4.1, called through
\texttt{qldpc} 0.2.6) applied to the full matrix $H_Z$. It runs belief
propagation with the minimum-sum rule, at most 30 iterations, and a
serial message schedule in a fixed qubit order, followed by
localized statistics decoding of order zero. The prior error
probability of every qubit is $\theta=1/(1+\e)$, so that every qubit
flip has unit cost $\log[(1-\theta)/\theta]=1$. We set $c(0)=0$. If the
decoder raises an error or returns an output that fails the check
$H_Zc(\sigma)=\sigma$, a fixed Gaussian-elimination solution is used
instead. The output is a deterministic function of $\sigma$ alone and
is the same for all values of $p$ and $q$.

One Monte Carlo step consists of $L^3$ stabilizer moves, one
single-qubit move per qubit in random order when $q>0$, and one logical
move. With parallel tempering, each step applies these updates at every
inverse temperature and then proposes exchanges of configurations
between neighboring inverse temperatures. Each replica is thermalized
for $2048$ steps, after which its logical class is recorded every $4$
steps. Depending on the parameter point, each replica records between
$4096$ and $196608$ classes, and between $2176$ and $14464$ disorder
realizations are used. We validated the implementation against exact
enumeration of the logical posterior for the seven-qubit Steane
code~\cite{steaneErrorCorrectingCodes1996}, with all seven nonzero
products of its three $Z$ generators measured as redundant checks.

We use the cubic filling geometry to bound the minimum correction
weights for these valid syndromes. If the total syndrome length is
smaller than the periodic linear size, each cycle in a decomposition
of the dual syndrome lifts to a closed path in the infinite cubic
lattice. Exchanging neighboring coordinate steps fills the path with
elementary plaquettes until opposite steps cancel. The number of
plaquettes required is bounded by a constant times the square of the
path length. Summing over the cycle decomposition then bounds the total
filling area by a constant times the square of $|\sigma|$.
This regime therefore admits a soundness function $f$ with
$f(|\sigma|)=O(|\sigma|^2)$.

To see why a uniform linear bound fails, consider the boundary of a
square dual membrane much smaller than the periodic linear size.
Its syndrome weight scales with the side length, whereas any filling
requires an area of the order of the square of that length. Projection
onto the supporting coordinate plane gives the lower bound on area.
For a square small compared with the torus, adding a noncontractible
surface cannot reduce this bound. Taking larger squares as the system
grows therefore rules out a uniform linear filling bound. This argument
concerns the $Z$-syndrome map detecting membrane-like $X$ errors.

\bibliography{ref}

\end{document}